\documentclass[conference]{IEEEtran}

\title{Axiomatizations and Computability of Weighted Monadic Second-Order Logic}

\author{\IEEEauthorblockN{Antonis Achilleos}
  
  \IEEEauthorblockA{Department of Computer Science\\
    Reykjavik University,
    Reykjavik, Iceland\\
    Email: antonios@ru.is}
  \and
  \IEEEauthorblockN{Mathias Ruggaard Pedersen}
  \IEEEauthorblockA{Department of Computer Science\\
    Reykjavik University,
    Reykjavik, Iceland\\
    Email: mathias.r.pedersen@gmail.com}
}

\usepackage[utf8]{inputenc}
\usepackage{amsmath}
\usepackage{amssymb}
\usepackage{amsthm}
\usepackage{stmaryrd} 
\usepackage{url}
\usepackage{complexity}
\usepackage{array}
\usepackage{relsize}
\usepackage{tabu}
\usepackage{xcolor}

\newtheorem{theorem}{Theorem}
\newtheorem{lemma}{Lemma}
\newtheorem{proposition}{Proposition}
\newtheorem{corollary}{Corollary}

\theoremstyle{remark}
\newtheorem{remark}{Remark}

\theoremstyle{definition}
\newtheorem{definition}{Definition}
\newtheorem{example}{Example}

\newcommand{\subparagraph}[1]{\paragraph*{#1}}

\newcommand{\true}{\top}
\newcommand{\false}{\bot}
\newcommand{\sat}[1]{\left\llbracket #1 \right\rrbracket}
\newcommand{\ra}{\rightarrow}
\newcommand{\nil}{\mathbf{0}}
\newcommand{\cond}{\mathbin{?}}

\newcommand{\multiset}[1]{\{\!| #1 |\!\}}
\newcommand{\multi}[1]{\mathbb{N}\multiset{#1}}
\newcommand{\munion}{\uplus}
\newcommand{\bigmunion}{\biguplus}
\newcommand{\aggr}{\mathtt{aggr}}

\newcommand{\swFO}{\textsf{step-wFO}}
\newcommand{\cwFO}{\textsf{core-wFO}}
\newcommand{\MSO}{\textsf{MSO}}
\newcommand{\swMSO}{\textsf{step-wMSO}}
\newcommand{\cwMSO}{\textsf{core-wMSO}}

\newcommand{\var}{\mathtt{var}}

\newcommand{\depth}{\mathtt{depth}}

\newcommand{\ncond}{\#?}

\newcommand{\eqeq}{\pmb{=}}
\newcommand{\defeq}{\stackrel{\textsf{def}}{=}}

\newcommand{\tend}{\ensuremath{{\triangleleft}}}
\newcommand{\mi}{\ensuremath{\texttt{m}}}
\newcommand{\cnt}{\ensuremath{\texttt{1}}}
\newcommand{\blank}{\_}

\newcommand{\sym}[1]{\ensuremath{\texttt{smbl}(#1)}}
\newcommand{\first}[1]{\ensuremath{\texttt{first}(#1)}}
\newcommand{\firstc}[1]{\ensuremath{\texttt{1st-cf}(#1)}}
\newcommand{\firstcx}[1]{\ensuremath{\texttt{1st-cf-x}(#1)}}
\newcommand{\last}[1]{\ensuremath{\texttt{last}(#1)}}
\newcommand{\nxt}[1]{\ensuremath{\texttt{nxt}(#1)}}
\newcommand{\nxtsymb}[1]{\ensuremath{\texttt{nxt-sm}(#1)}}
\newcommand{\valone}[1]{\ensuremath{v_1(#1)}}
\newcommand{\valv}[2]{\ensuremath{v_{#2}^{#1}}}
\newcommand{\valzer}{\ensuremath{\nil}}
\newcommand{\pos}[2]{\ensuremath{\texttt{ps}_{#2}(#1)}}
\newcommand{\tr}[1]{\ensuremath{\texttt{tr}(#1)}}
\newcommand{\trprime}[1]{\ensuremath{\texttt{tr'}(#1)}}

\newcommand{\eqform}{\chi}
\newcommand{\swform}{\Psi}
\newcommand{\cwform}{\Phi}

\newcommand{\wgt}{\texttt{wgt}}
\newcommand{\nat}{\mathbb{N}}

\newcommand{\sep}{,}

\begin{document}

\IEEEoverridecommandlockouts
\IEEEpubid{\makebox[\columnwidth]{978-1-6654-4895-6/21/\$31.00~
    \copyright2021 IEEE \hfill} \hspace{\columnsep}\makebox[\columnwidth]{ }}

\maketitle

\begin{abstract}
  Weighted monadic second-order logic is a weighted extension
  of monadic second-order logic 
  that
  captures exactly the behaviour of weighted automata. 
  Its semantics is parameterized with respect to a semiring on which the values that weighted formulas output are evaluated.
  Gastin and Monmege (2018) gave abstract semantics for a version of weighted monadic second-order logic to give a more general and modular proof of the equivalence of the logic with weighted automata.
  We focus on the abstract semantics of the logic and we 
  give a complete axiomatization both for the full logic
  and for a fragment without general sum,
  thus giving a more fine-grained understanding of the logic.
  We discuss how common decision problems for logical languages
  can be adapted to the weighted setting,
  and show that many of these are decidable,
  though they inherit bad complexity
  from the underlying first- and second-order logics.
  However, we show that a weighted adaptation of satisfiability is undecidable
  for the logic when one uses the abstract interpretation.
\end{abstract}

\begin{IEEEkeywords}
  weighted logic\sep~ monadic second-order logic\sep~ axiomatization\sep~ weighted automata\sep~ satisfiability
\end{IEEEkeywords}

\section{Introduction}

Weighted logics are a quantitative generalization of classical logics
that allows one to reason about quantities such as probabilities,
cost, production, or energy consumption in systems \cite{DG07,GM18,KR13}.
These kinds of logic are important,
since they allow us to describe not only that, 
for example, 
a certain task was completed,
but also that only a specific amount of resources
were consumed in order to complete the task.
One of the main results of the theory of weighted logics
is the correspondence between weighted monadic second-order logic \cite{DG07}
or quantitative monadic second-order logic \cite{KR13} and weighted automata,
thus generalizing the classical result of B\"{u}chi, Elgot, and Trakhtenbrot \cite{B60,E61,T61},
of the equivalence  
between
classical finite automata and monadic second-order logic (MSO). 
This is important because it shows that weighted or quantitative 
MSO
are well-suited to reason about weighted automata,
which themselves are a popular tool for modeling systems,
having found applications in areas such as image compression \cite{CK93}
and natural language processing \cite{KM09,MPR96}.
The correspondence between weighted 
MSO
and weighted automata has been adapted to many other computational models,
such as weighted Muller tree automata \cite{R07}
and weighted picture automata \cite{F11}.

Complete axiomatizations for weighted logics
as well as their decision problems have 
been well-studied
in the context of weighted extensions of modal logics
for weighted transition systems. 
Larsen and Mardare \cite{LM14} gave a complete axiomatization for weighted modal logic on weighted transition systems%
,
and they later extended this work to also consider concurrency \cite{LMX15a}.
Hansen et al. 
gave
a
complete axiomatization 
for a logic
to reason about bounds in weighted transitions systems in \cite{HLMP18},
where
they also show
the decidability of the satisfiability problem. 
Larsen  et al. proved in \cite{LMX18} that
the satisfiability problem 
is
decidable for a weighted logic
with recursion, in which recursive equations can describe
infinite behaviour.
Similarly, Larsen et al. gave in \cite{LMX15b} both a complete axiomatization and a decision procedure
for satisfiability 
for the alternation-free fragment of
a weighted extension of the $\mu$-calculus on weighted transition systems.

In contrast to these weighted extensions of modal logics, and to 
MSO
and first-order logic (FO),
which have been well-studied for decades,
there has not been a study of 
the 
weighted extensions 
of MSO and FO
from the axiomatic point of view.
In this paper we initiate this study by giving axiomatisations
of 
the equational theory of
fragments of weighted 
MSO, 
as formulated in \cite{GM18}, 
and considering
the decidability of some of its decision problems.
The weighted variation
\cwMSO\ of MSO and its abstract interpretation  was proposed by Gastin and Monmege in \cite{GM18} to prove a general result about the correspondence between weighted MSO and weighted automata.
The approach of Gastin and Monmege is modular, in that 
both the syntax and the semantics of \cwMSO\ is given in layers.
A formula or automaton is first given an abstract interpretation, by returning a multiset of strings of weights; and then an aggregator function maps that multiset into the desired structure of values --- typically a semiring. 
By proving the correspondence of formulas and automata, the result of Gastin and Monmege does not depend upon the specific structure of the weights.
Furthermore, 
\cwMSO\ has three syntactic layers, each with different characteristics.

\subsubsection*{Our contribution}
We give three complete axiomatizations: 
one for the full second syntactic layer
of weighted MSO, one for a fragment of 
the third and final syntactic layer,
and one for the full third layer.
Each of the three axiomatizations exhibits different 
characteristics and machinery, and therefore our presentation of the axioms allows for axiomatizations that are taylored to each fragment of \cwMSO.
Due to the modular nature of \cwMSO, these axiomatizations also apply to \cwFO, the first-order version of the logic.
We prove that the equivalence problem for weighted automata under the abstract interpretation, can be solved in polynomial time (Corollary \ref{cor:automataP}).
We also show that the model checking, satisfiability, and validity problems,
appropriately translated to the equational, weighted setting,
are decidable for the second layer of the logic,
although these inherit the 
\PSPACE-completeness and non-elementary complexity
from MSO
for model checking and, respectively,  
for satisfiability and validity.
However, for the third layer of the logic, things are more complicated.
The model checking and validity problems remain decidable,
but we show that the satisfiability problem is undecidable, even for the first-order fragment.

\subsubsection*{Related work}
Weighted MSO was introduced by Droste and Gastin in \cite{DG07}.
The version that we study in this paper was defined by Gastin and Monmege in \cite{GM18}, where they prove that it is equivalent to weighted automata.
Droste and Gastin defined a number of first-order restrictions of the logic in \cite{DG19}, where they prove correspondence with corresponding restrictions of weighted automata.
Naturally, more variations of weighted MSO have appeared, for example, extending the logic with multiple weights \cite{DP16}, or on infinite words \cite{DM12}.
To the best of our knowledge, the present paper is the first to give an axiomatization for \cwMSO.
One can find several results about the decidability and complexity of problems about weighted automata in the literature, and these tend to vary, depending on the structure of the weights.
Already from \cite{S61}, Sch\"{u}tzenberger proves that determining the equivalence of $(\mathbb{Q},+,\cdot)$-weighted automata can be done in polynomial time, and we show the same result for the abstract semantics using a different proof.
On the other hand, the same problem over the $(\mathbb{Q},\max,+)$ semiring is undecidable \cite{ABK20,Krob1994}.
Droste and Gastin in \cite{DG07} show that the equivalence problem for formulas --- which we call equational validity in this paper --- over computable, commutative, and locally finite semirings, is decidable.
This and other decidability results (for example, \cite{DP16} and \cite{DM12}) for weighted MSO result from the translation of formulas to automata.
This paper provides an alternative method to decide equational validity, by a proof system.
A problem similar to this paper's equational satisfiability is proven undecidable by Bollig and Gastin in \cite{BG09} for probabilistic logics on trees.
To the best of our knowledge, this paper is the first effort to tackle the decidability of equational satisfiability and validity of \cwMSO\ over the abstract semantics.

%
%


\section{Preliminaries}
Given a set $X$, we denote by $X^*$ the set of words over $X$,
and by $X^+$ the set of non-empty words over $X$.
Given a word $w$, we denote by $|w|$ the length of $w$.

Denote by $\multi{X}$ the collection of all finite multisets over $X$,
where a finite multiset is a function $f : X \ra \mathbb{N}$
such that $f(x) \neq 0$ for finitely many $x \in X$.
Intuitively, $f(x)$ tells us how many times the element $x$ occurs in the multiset $f$.
We will use $\multiset{\cdot}$ to denote a multiset,
so that e.g. $\multiset{1,1,2,3}$ is the multiset that contains two $1$'s, one $2$, and one $3$.
The union $\munion$ of two multisets $f$ and $g$ is defined pointwise as
\((f \munion g)(x) = f(x) + g(x).\)

A semiring is a tuple $(X, +, \times, 0, 1)$ such that
$(X, \times, 1)$ is a monoid ($\times$ is an associative binary operation on $X$, with $1$ as an identity element),
$(X, +, 0)$ is a commutative monoid (it is a monoid and $+$ is commutative),
$\times$ distributes over $+$,
and $0 \times x = x \times 0 = 0$ for all $x \in X$.
Some common examples of semirings are $(\mathbb{Z}, +, \times, 0, 1)$,
the integers with the usual sum and product,
and $(\{0,1\}, \lor, \land, 0, 1)$, the Boolean semiring with the usual
Boolean disjunction and conjunction.
For our purposes, another important example of semirings
is that of $(\multi{X^*}, \munion, \cdot, \emptyset, \multiset{\varepsilon})$,
the semiring over multisets of sequences over $X$, with multiset union as sum,
concatenation as product, the empty set as zero,
and the multiset containing only the empty string once as identity.

\section{Syntax and semantics}
Our presentation of weighted MSO
follows the style of \cite{GM18},
in which the logic is separated into three different layers.
The first layer is simply MSO.
The second layer is called \swMSO\ and
is built from single values and 
if-then-else statements with
MSO formulas as conditions.
The name of this comes from the fact that
its semantics describes step functions,
that is, functions that return values from a fixed finite set of weights.
The third and last layer is called \cwMSO\  and
allows one to combine the single values from the second layer into
more complex expressions using sums and products.

We use
a countably infinite set of first-order variables $\mathcal{V}_{FO}$,
a countably infinite set of second-order variables $\mathcal{V}_{SO}$,
a finite alphabet $\Sigma$, and an arbitrary set $R$ of weights.
The syntax of weighted MSO is given by the following grammar.

\noindent \MSO:
  \[\varphi ::= \true \mid P_a(x) \mid x \leq y \mid x \in X \mid \neg \varphi \mid \varphi_1 \land \varphi_2 \mid \forall x. \varphi \mid \forall X. \varphi\]
  \swMSO:
  \[\swform ::= r \mid \varphi \cond \swform_1 : \swform_2\]
  \cwMSO:
  \[\cwform ::= \nil \mid \textstyle{\prod_x} \swform \mid \varphi \cond \cwform_1 : \cwform_2 \mid \cwform_1 + \cwform_2 \mid \textstyle{\sum_x \cwform} \mid \textstyle{\sum_X \cwform}\]
where $a \in \Sigma$, $r \in R$, $x,y \in \mathcal{V}_{FO}$, and $X \in \mathcal{V}_{SO}$.
In the rest of the paper, we use $\varphi$ to denote \MSO\ formulas, $\swform$ to denote \swMSO\ formulas,  $\cwform$ to denote \cwMSO\ formulas, and $\eqform$ to denote \swMSO\ or \cwMSO\ formulas.

In a similar fashion, we obtain \swFO\ by only allowing conditioning on first-order formulas in \swMSO\
and \cwFO\ by only allowing first-order formulas and removing the construct $\sum_X \cwform$
which sums over a second-order variable.

We will use \valzer\ as the default (negative) value in conditionals,
and as such $\varphi \cond \cwform$ is used as shorthand for $\varphi \cond \cwform : \valzer$.
Furthermore, we assume that $:$ binds to the nearest $\cond$, and therefore, $\varphi_1 \cond \varphi_2 \cond \cwform_1 : \cwform_2$ means $\varphi_1 \cond \varphi_2 \cond \cwform_1 : \cwform_2 : \valzer$, which can be uniquely parsed as $\varphi_1 \cond (\varphi_2 \cond \cwform_1 : \cwform_2) : \valzer$.
For a \swMSO\ formula $\swform$, $R(\swform) = \{ r \in R \mid r $ appears in $\swform \}$; for brevity, we may write $r \in \swform$ instead of $r \in R(\swform)$.

We note here that in earlier work of Droste and Gastin \cite{DG07},
a different formulation of weighted MSO was given.
There, the syntax was essentially the same as the syntax for
classical monadic second-order logic,
and the semantics were given as a function from words and valuations to elements of $R$.
However, one can translate between the formulation presented here
and a restricted version of the formulation of \cite{DG07},
as was shown in \cite[Section 5]{GM18}.
We choose to follow the formulation of \cite{GM18}
because this gives a cleaner correspondence with weighted automata,
whereas the earlier formulation of weighted MSO required a (not fully syntactic)
restriction in order to obtain a correspondence with weighted automata,
and because the abstract semantics of this formulation allows us to
focus on the syntactic level, which is ideal for an axiomatization.

The formulas $\varphi$ of \MSO\ are interpreted over words $w \in \Sigma^+$
together with a valuation $\sigma$ of this word,
which assigns to each first-order variable a position in the word
and to each second-order variable a set of positions in the word.

When interpreted on a string, a formula outputs a value, which, concretely, may be a single weight, a sequence, or, say, a set (or multiset) of more elementary values.
To preserve the generality of the logic, the semantics are given in two steps.
The first is an abstract semantics, where the meaning of a formula
is given as a multiset of sequences of weights.
The second is a concrete semantics, where one can translate the abstract semantics
into a given semiring structure, by assuming an appropriate operator on the abstract values.

We denote by $\Sigma^+_{val}$ the set of pairs $(w,\sigma)$ where $w \in \Sigma^+$
and $\sigma$ is a valuation of $w$.
Let $x$ be a first-order (respectively, let $X$ be a second-order) variable and $i \in \{ 1, \dots, |w| \}$ (respectively, $I \subseteq \{ 1, \dots, |w| \}$).
By $\sigma[x \mapsto i]$ (respectively $\sigma[X \mapsto I]$) we denote the valuation that maps each variable $y$ and $Y$ to $\sigma(y)$ and $\sigma(Y)$, if $y \neq x$ (respectively, if $Y \neq X$), and $x$ to $i$ (respectively, $X$ to $I$).
The semantics of \MSO\ on finite words is standard and can be found in e.g. 
\cite{L04}.
%
In this paper, $(w,\sigma)$ will always be a pair from $\Sigma^+_{val}$.

We denote by $\sat{\varphi}$ the set of all pairs $(w,\sigma) \in \Sigma^+_{val}$
that satisfy $\varphi$.
Likewise, for a set $\Gamma$ of \MSO\ formulas,
we define
\[\sat{\Gamma} = \begin{cases} \Sigma^+_{val} & \text{if } \Gamma = \emptyset \\ \bigcap_{\varphi \in \Gamma} \sat{\varphi} & \text{otherwise}. \end{cases}\]
The semantics of formulas $\swform$ of \swMSO\ is given
by a function $\sat{\cdot} : \Sigma^+_{val} \ra R$,
defined by $\sat{r}(w,\sigma) = r$ and
\[
\sat{\varphi \cond \swform_1 : \swform_2}(w,\sigma) = \begin{cases} \sat{\swform_1}(w,\sigma) \text{ if } (w,\sigma) \models \varphi \\ \sat{\swform_2}(w,\sigma) \text{ otherwise.}\end{cases}
\]
The semantics of formulas $\cwform$ of \cwMSO\
is given by the function $\sat{\cdot} : \Sigma^+_{val} \ra \multi{R^*}$:
\begin{align*}
  \sat{\nil}(w,\sigma) &= \emptyset \\
  \sat{\textstyle{\prod_x} \swform}(w,\sigma) &= \multiset{r_1 r_2 \dots r_{|w|}}, r_i = \sat{\swform}(w,\sigma[x {\mapsto} i]) \\
  \sat{\varphi \cond \cwform_1 : \cwform_2}(w,\sigma) &= \begin{cases} \sat{\cwform_1}(w,\sigma), \text{ if } (w,\sigma) \models \varphi \\ \sat{\cwform_2}(w,\sigma), \text{ otherwise}\end{cases} \\
  \sat{\cwform_1 + \cwform_2}(w,\sigma) &= \sat{\cwform_1}(w,\sigma) \munion \sat{\cwform_2}(w,\sigma) \\
  \sat{\textstyle{\sum_x} \cwform}(w,\sigma) &= \bigmunion_{i \in \{1, \dots, |w|\}} \sat{\cwform}(w,\sigma[x \mapsto i]) \\
  \sat{\textstyle{\sum_X} \cwform}(w,\sigma) &= \bigmunion_{I \subseteq \{1, \dots, |w|\}} \sat{\cwform}(w, \sigma[X \mapsto I])
\end{align*}

Let $\Gamma$ be a set of \MSO\ formulas.
We say that two formulas $\eqform_1$ and $\eqform_2$
are semantically $\Gamma$-equivalent and write
$\eqform_1 \sim_\Gamma \eqform_2$
if $\sat{\eqform_1}(w,\sigma) = \sat{\eqform_2}(w,\sigma)$
for all $(w,\sigma) \in \sat{\Gamma}$.
If $\Gamma = \emptyset$, we simply write $\eqform_1 \sim \eqform_2$
and say that $\eqform_1$ and $\eqform_2$
are semantically equivalent.

\subparagraph{Concrete semantics}
To obtain the concrete semantics of a formula for a given semiring structure $(X, +, \times, 0 , 1)$,
we assume an \emph{aggregation function} $\aggr : \multi{R^*} \rightarrow X$.
Note that the set $X$ may be different from the set of weights $R$.

\begin{example}\label{ex:1}
  Let $\Sigma = \{a,b\}$, $R = \{0,1\}$ and consider the max-plus semiring
  $(\mathbb{N} \cup \{-\infty\}, \max, +, -\infty, 0)$.
  We wish to count the maximum number of consecutive $a$'s in a given string $w \in \Sigma^*$.
  We define the aggregation function as
  $
  \aggr(M) = \max_{r_1 \dots r_n \in M}( r_1 + \dots + r_n),
  $
  thus interpreting the sum and product of the multiset sequence semiring ($\munion$ and $\cdot$)
  as the corresponding sum and product ($\max$ and $+$) in the max-plus semiring.
  Now define the first-order formula $\varphi$ as
  $
  \varphi = x \leq y \land \forall z.((x \leq z \land z \leq y) \rightarrow P_a(z)),
  $
  and let
  $\swform = \varphi \cond 1 : 0$, 
  $\cwform' = \textstyle{\prod_y} \swform$, 
  	and 
	$\cwform = \textstyle{\sum_x} \cwform'$,
  so that $\cwform = \sum_x \prod_y \varphi \cond 1 : 0$.
  Consider the string $w = abaa$, which has a maximum number of two consecutive $a$'s.
  We find that
  \begin{align*}
    \sat{\cwform}(w,\sigma) &= \sat{\cwform'}(w,\sigma[x \mapsto 1]) \munion \sat{\cwform'}(w, \sigma[x \mapsto 2]) \\
    &\phantom{{}={}}\munion \sat{\cwform'}(w,\sigma[x \mapsto 3]) \munion \sat{\cwform'}(w, \sigma[x \mapsto 4]) \\
    &= \multiset{1000} \munion \multiset{0000} \munion \multiset{0011} \munion \multiset{0001} \\
    &= \multiset{1000, 0000, 0011, 0001}
  \end{align*}
  and hence the concrete semantics become
  \begin{align*}
    \aggr(\sat{\cwform}(w,\sigma)) &= \aggr(\multiset{1000, 0000, 0011, 0001}) \\
    &= \max\{1,0,2,1\} = 2,
  \end{align*}
  which is 
  the maximum number of 
  consecutive $a$'s in $w$.
\end{example}

Semiring semantics have some limitations in their expressive power,
and some natural quantities, such as discounted sum,
can not be computed using these semantics.
Alternative concrete semantics have therefore been proposed that
give more expressive power, such as valuation monoids \cite{DM12}
and valuation structures \cite{DP16},
which allows one to compute more complex quantities,
such as optimal discounted cost, average of ratios, and more.
In this paper, we work exclusively with abstract semantics.

\begin{remark}
	It is important to note that the abstract semantics that were defined in this section can be seen as a kind of concrete semantics, for the semiring structure $(\multi{R^*}, \munion, \cdot, \emptyset, \multiset{\varepsilon})$, 
	the semiring over multisets of sequences over the weights.
\end{remark}

\section{Decision problems}

The
three 
usual
decision problems 
that one considers for 
a logical language
are
model checking, satisfiability, and validity.
The model checking problem asks if a given model satisfies a given formula,
the satisfiability problem asks whether for a given formula
there exists a model that satisfies the formula,
and the validity problem asks if a given formula is satisfied by all models.
For \FO, \MSO, and many classical Boolean logics, the satisfiability and validity problems are equivalent,
since the satisfiability of a formula $\varphi$
is equivalent to the non-validity of its negation $\neg \varphi$.

In this section, we briefly discuss how to extend these fundamental notions to
our setting of a non-Boolean logic. 
We assume the set $R$ of weights has decidable equality,
i.e. it is decidable (with reasonable efficiency, when discussing complexity issues) whether $r_1 = r_2$ for two weights $r_1,r_2 \in R$.
First, observe that we can encode every \MSO\ formula as an equation (Lemma \ref{lem:embedMSOtoEq}), and vice-versa (Lemma \ref{lem:WtoMSO}).
\begin{lemma}\label{lem:embedMSOtoEq}
	Assume two distinct values, $0,1 \in R$, and let $\varphi \in \MSO$.
	Then, for every $(w,\sigma)$, the following are equivalent:
	\textbf{(1)} $(w,\sigma) \models \varphi$, 
	\textbf{(2)} $\sat{\varphi \cond 0 : 0}(w,\sigma) = \sat{\varphi \cond 0 : 1}(w,\sigma)$, and
	\textbf{(3)} $\sat{\varphi \cond \Pi_x 0 : \Pi_x 0}(w,\sigma) = \sat{\varphi \cond \Pi_x 0 : \Pi_x 1}(w,\sigma)$.
\end{lemma}

\begin{definition}
	For $\Psi \in \swMSO$ and $r \in R$, we define $\varphi(\Psi,r)$ recursively:
	\begin{itemize}
		\item $\varphi(r,r) = \top$ and $\varphi(r',r) = \neg \top$, when $r \neq r'$; and
		\item 
		$\varphi(\varphi' \cond \Psi_1 : \Psi_2,r) = (\varphi' \land \varphi(\Psi_1,r)) \lor (\neg\varphi' \land \varphi(\Psi_2,r))$.
	\end{itemize}
\end{definition}

\begin{lemma}\label{lem:WtoMSO}
	$(w,\sigma) \in \sat{\varphi(\Psi,r)}$ iff $\sat{\Psi}(w,\sigma) = r$.
\end{lemma}

We consider weighted model checking, which has two versions. We recall that for \MSO\ and \FO, model checking is \PSPACE-complete \cite{S74,V82}.

\begin{definition}[The evaluation problem] Given $(w,\sigma)$ and a formula $\eqform$,
	compute $\sat{\eqform}(w,\sigma)$.
\end{definition}

\begin{definition}[Weighted model checking problem] Given $(w,\sigma)$, a formula $\eqform$,
and a weight or multiset $v$, do we have $\sat{\eqform}(w,\sigma) = v$?
\end{definition}

To evaluate a $\swMSO$ or $\cwMSO$ formula on $(w,\sigma)$, one can use the recursive procedure that is yielded by the semantics of $\swMSO$ and $\cwMSO$, using a
model checking algorithm for $\MSO$ to check which branch to take at each conditional.
It is not hard to see that for \swMSO, this can be done using polynomial space, as that fragment only uses conditionals on \MSO\ formulas and values.
%
Then, using Lemmata \ref{lem:embedMSOtoEq} and \ref{lem:WtoMSO}: 
\begin{theorem}
The weighted model checking problem is 
\PSPACE-complete for \swMSO.
\end{theorem}

Next we consider several variations of the satisfiability problem in the weighted setting.

%

\begin{definition}[$r$-satisfiability] Given $\eqform$ and a weight or multiset $v$, is there $(w,\sigma)$ such that $\sat{\eqform}(w,\sigma) = v$?
\end{definition}

For $\swMSO$ formulas, this 
problem has the same complexity as \MSO\ satisfiability over finite strings, using Lemmata \ref{lem:WtoMSO} and \ref{lem:embedMSOtoEq}.
Therefore, the problem is decidable, but with a nonelementary complexity \cite{R02}.
For $\cwMSO$ formulas $\cwform$, 
the problem 
is similar to the following variation.

\begin{definition}[Equational satisfiability] Given $\eqform_1$ and $\eqform_2$,
does there exist $(w,\sigma)$ such that $\sat{\eqform_1}(w,\sigma) = \sat{\eqform_2}(w,\sigma)$?
\end{definition}

For $\swMSO$ formulas, this problem is decidable in the same way as $r$-satisfiability, by reducing to the satisfiability problem of \MSO: 
there exist $(w,\sigma)$ such that $\sat{\swform_1}(w,\sigma) = \sat{\swform_2}(w,\sigma)$ if and only if the following formula is satisfiable: 
\[\bigvee_{
	r\in R(\swform_1) \cap R(\swform_2)
	}\varphi(\swform_1,r) \land \varphi(\swform_2,r).\]
For $\cwMSO$, and even $\cwFO$ formulas,
we show that this problem is undecidable in Section \ref{sec:undec} (Theorem \ref{thm:sat-is-undec}).
Finally, we consider a version of 
validity in the weighted equational setting.

\begin{definition}[Equational validity] Given $\eqform_1$ and $\eqform_2$,
do we have $\sat{\eqform_1}(w,\sigma) = \sat{\eqform_2}(w,\sigma)$ for all $(w,\sigma)$?
\end{definition}

This problem is decidable. 
As we give in Section \ref{sec:completeness} (Theorem \ref{thm:completeness-full}) a recursive and complete axiomatization of the equational theory of \cwMSO, 
the problem is recursively enumerable (\RE). But the logic also has a recursive set of models and a decidable evaluation problem. Therefore, this version of validity is also \coRE, and therefore decidable.


\section{Axioms}\label{sec:completeness}
Just as the syntax of the logic was given in three layers,
we 
also present the 
axioms of the logic
in three layers,
one for each of the syntactic layers.
We note that the proofs of completeness do not rely on any properties
of \MSO\ itself, apart from it having a complete axiomatization,
and therefore the axiomatizations also apply to $\swFO$ and $\cwFO$.

For the $\swMSO$ and $\cwMSO$ layers, which are 
not
Boolean,
we give an axiomatization in terms of equational logic. 
For a set $\Gamma$ of $\MSO$ formulas,
we 
use the notation
$\Gamma \vdash \eqform_1 \approx \eqform_2$ 
to mean that
under the assumptions in $\Gamma$, $\eqform_1$ is equivalent to $\eqform_2$.
These equations
must satisfy the axioms of equational logic,
which are reflexivity, symmetry, transitivity, and congruence,
as reported in Table \ref{tab:eq-axioms}.
Note that the congruence rule for sum, \textsf{cong+},
only applies to the $\cwMSO$ layer, since $\swMSO$ has no sum operator.
Furthermore, the congruence rule for the conditional operator, \textsf{cong?},
is not strictly necessary to include,
since it can be derived from the 
axioms that we introduce later.
However, to follow standard presentations of equational logic,
we include it as part of the axioms here.

\begin{table}
  \centering
  \begin{tabular}{r l}
    \hline
    (\textsf{ref}): & $\Gamma \vdash \eqform \approx \eqform$ \\[1ex]
    (\textsf{sym}): & $\Gamma \vdash \eqform_1 \approx \eqform_2$ implies $\Gamma \vdash \eqform_2 \approx \eqform_1$ \\[1ex]
    (\textsf{trans}): & $\begin{aligned}&\Gamma \vdash \eqform_1 \approx \eqform_2 \text{ and } \Gamma \vdash \eqform_2 \approx \eqform_3 \\[-1ex]
      &\text{implies } \Gamma \vdash \eqform_1 \approx \eqform_3 \end{aligned}$ \\[2ex]
    (\textsf{cong?}): & $\begin{aligned}&\Gamma \vdash \eqform_1 \approx \eqform_1' \text{ and } \Gamma \vdash \eqform_2 \approx \eqform_2' \\[-1ex]
      &\text{implies } \Gamma \vdash \varphi \cond \eqform_1 : \eqform_2 \approx \varphi \cond \eqform_1' : \eqform_2' \end{aligned}$ \\[2ex]
    (\textsf{cong+}): & $\begin{aligned}&\Gamma \vdash \eqform_1 \approx \eqform_1' \text{ and } \Gamma \vdash \eqform_2 \approx \eqform_2' \\[-1ex]
      &\text{implies } \Gamma \vdash \eqform_1 + \eqform_2 \approx \eqform_1' + \eqform_2' \end{aligned}$ \\
    \hline \\
  \end{tabular}
  \caption{Axioms for equational logic.}
  \label{tab:eq-axioms}
\end{table}


\subsection{\MSO}
\MSO\ over finite strings is equivalent to finite automata  \cite{B60,E61,T61}, and therefore it also has a decidable validity problem (albeit with a nonelementary complexity). This means that the theory of \MSO\ over finite strings has a recursive and complete axiomatization. One such axiomatization is given in \cite{GC12}, and therefore for a set $\Gamma \cup \{ \varphi \}$ of \MSO\ formulas, 
$\Gamma \vdash \varphi$ means that $\varphi$ is derivable from these axioms and $\Gamma$ ($\Gamma$ may be omitted when empty).
Since \FO\ over finite strings also has a decidable validity problem,
it likewise has a recursive and complete axiomatization.
For the purpose of this paper, we fix one such axiomatization,
and we can thus also write $\Gamma \vdash \varphi$ when $\Gamma \cup \{\varphi\}$
is a set of \FO\ formulas.

\begin{theorem}[Completeness of \MSO\ \cite{GC12}]\label{thm:completeness-mso}
  For every \MSO\ formula $\varphi$,
  $\models \varphi$ if and only if $\vdash \varphi$.
\end{theorem}

\begin{corollary}\label{cor:completeness-mso}
  For every finite $\Gamma$,
  $\Gamma \models \varphi$ if and only if $\Gamma \vdash \varphi$.
\end{corollary}
%
%

\subsection{\swMSO}
The equational axioms for $\swMSO$ are given in Table \ref{tab:swmso-axioms}.
Axiom $(S1)$ allows one to add additional assumptions to $\Gamma$,
and $(S2)$ shows how negation affects the conditional operator
by switching the order of the results.
Axiom $(S3)$ shows that if the formula $\varphi$ that is being conditioned on
can be derived from $\Gamma$ itself,
then the first choice of the conditional will always be taken.
Finally, $(S4)$ gives a way to remove assumptions
and put them into a conditional statement instead:
If the first choice of the conditional is equivalent to $\swform$
under the assumption that $\varphi$ is true,
and the second choice of the conditional is equivalent to $\swform$
under the assumption that $\varphi$ is false,
then the conditional is equivalent to $\swform$.

\begin{table}
  \centering
  \begin{tabular}{r l}
    \hline
    ($S1$): & $\Gamma \vdash \swform_1 \approx \swform_2$ implies $\Gamma \cup \{\varphi\} \vdash \swform_1 \approx \swform_2$ \\[1ex]
    ($S2$): & $\Gamma \vdash \neg \varphi \cond \swform_1 : \swform_2 \approx \varphi \cond \swform_2 : \swform_1$ \\[1ex]
    ($S3$): & $\Gamma \vdash \varphi$ implies $\Gamma \vdash \varphi \cond \swform_1 : \swform_2 \approx \swform_1$ \\[1ex]
    ($S4$): & $\begin{aligned}&\Gamma \cup \{\varphi\} \vdash \swform_1 \approx \swform \text{ and } \Gamma \cup \{\neg \varphi\} \vdash \swform_2 \approx \swform \\[-1ex]
    &\text{implies } \Gamma \vdash \varphi \cond \swform_1 : \swform_2 \approx \swform \end{aligned}$ \\
    \hline \\
  \end{tabular}
  \caption{Axioms for \swMSO.}
  \label{tab:swmso-axioms}
\end{table}

Before proving that the axioms given in Table \ref{tab:swmso-axioms} are complete,
we first give some examples of theorems that can be derived from the axioms,
some of which will be used in the proof of completeness.
The first two of these are 
particularly
interesting,
since they give properties that are common in many logical systems,
namely the principle of explosion and the cut elimination rule.
The remaining theorems show that the conditional operator
behaves 
as expected,
and that all of these behaviours can be inferred
from the four axioms of Table \ref{tab:swmso-axioms}.

\begin{proposition}\label{prop:swmso-theorems}\label{prop:cwmso-theorems}
  The following theorems can be derived in \swMSO.
  \begin{enumerate}
    \item $\Gamma \vdash \swform_1 \approx \swform_2$ for any
      $\swform_1$ and $\swform_2$ if $\Gamma$ is inconsistent.
    \item $\Gamma \vdash \swform_1 \approx \swform_2$ if $\Gamma \vdash \varphi$ and $\Gamma \cup \{\varphi\} \vdash \swform_1 \approx \swform_2$.
    \item $\Gamma \vdash \varphi \cond \swform : \swform \approx \swform$.
    \item If $\Gamma \cup \{\varphi_1,\varphi_2\} \vdash \swform_1 \approx \swform_1'$,
      $\Gamma \cup \{\varphi_1, \neg \varphi_2\} \vdash \swform_1 \approx \swform_2'$,
      $\Gamma \cup \{\neg \varphi_1, \varphi_2\} \vdash \swform_2 \approx \swform_1'$, and
      $\Gamma \cup \{\neg \varphi_1, \neg \varphi_2\} \vdash \swform_2 \approx \swform_2'$, then
      $\Gamma \vdash \varphi_1 \cond \swform_1 : \swform_2 \approx \varphi_2 \cond \swform_1' : \swform_2'$.
    \item $\Gamma \vdash \varphi_1 \cond \swform_1 : \swform_2 \approx \varphi_2 \cond \swform_1 : \swform_2$
      if $\Gamma \vdash \varphi_1 \leftrightarrow \varphi_2$.
    \item $\Gamma \vdash \varphi \cond \swform_1 : \swform_2 \approx \swform_2$
      if $\Gamma \vdash \neg \varphi$.
    \item If $\Gamma \cup \{\varphi\} \vdash \swform_1 \approx \swform_1'$ and
      $\Gamma \cup \{\neg \varphi\} \vdash \swform_2 \approx \swform_2'$ then
      $\Gamma \vdash \varphi \cond \swform_1 : \swform_2 \approx \varphi \cond \swform_1' : \swform_2'$.
    \item If $\Gamma \cup \{\varphi\} \vdash \swform_1 \approx \swform_2$ and
      $\Gamma \cup \{\neg \varphi\} \vdash \swform_1 \approx \swform_2$
      then $\Gamma \vdash \swform_1 \approx \swform_2$.
    \item $\Gamma \cup \{\varphi\} \vdash \varphi \cond \swform_1 : \swform_2 \approx \swform_1$.
  \end{enumerate}
\end{proposition}
\begin{proof}
We  only prove some of these claims for illustration.
  \begin{enumerate}
    \item[1) ] Let $\swform_1$ and $\swform_2$ be arbitrary \swMSO\ formulas
    and assume that $\Gamma$ is inconsistent.
    Then $\Gamma \vdash \varphi$
    and $\Gamma \vdash \neg \varphi$.
    Then axiom ($S3$) gives
    $\Gamma \vdash \varphi \cond \swform_1 : \swform_2 \approx \swform_1$ and $\Gamma \vdash \neg \varphi \cond \swform_2 : \swform_1 \approx \swform_2$.
    Since $\Gamma \vdash \varphi \cond \swform_1 : \swform_2 \approx \neg \varphi \cond \swform_2 : \swform_1$ by axiom ($S2$), this implies $\Gamma \vdash \swform_1 \approx \swform_2$.
    \item[4)] Using ($S4$), $\Gamma \cup \{\varphi_1, \varphi_2\} \vdash \swform_1 \approx \swform_1'$
    and $\Gamma \cup \{\varphi_1, \neg \varphi_2\} \vdash \swform_1 \approx \swform_2'$ gives
    $\Gamma \cup \{\varphi_1\} \vdash \varphi_2 \cond \swform_1' : \swform_2' \approx \swform_1$.
    Likewise, using the other two assumptions, we get
    $\Gamma \cup \{\neg \varphi_1\} \vdash \varphi_2 \cond \swform_1' : \swform_2' \approx \swform_2$,
    and a final application of ($S4$) then gives
    $\Gamma \vdash \varphi_1 \cond \swform_1 : \swform_2 \approx \varphi_2 \cond \swform_1' : \swform_2'$.
    \item[5)] Assume that $\Gamma \vdash \varphi_1 \leftrightarrow \varphi_2$.
    Then, because of reflexivity and since $\{\varphi_1, \neg \varphi_2\}$ and $\{\neg \varphi_1, \varphi_2\}$ are inconsistent under $\Gamma$, the fourth item of this proposition gives that $\Gamma \vdash \varphi_1 \cond \swform_1 : \swform_2 \approx \varphi_2 \cond \swform_1 : \swform_2$.
    \qedhere
  \end{enumerate}
\end{proof}

The proof of completeness is
by a case analysis and induction on
the structure of the two formulas $\swform_1$ and $\swform_2$.
Lemma \ref{lem:completeness-swmso}
covers
the case where
both 
sides of the equation
are conditional statements.

\begin{lemma}\label{lem:completeness-swmso}
  If $\varphi_1 \cond \swform_1' : \swform_1'' \sim_\Gamma \varphi_2 \cond \swform_2' : \swform_2''$, then
  \[\begin{array}{l l} \swform_1' \sim_{\Gamma \cup \{\varphi_1, \varphi_2\}} \swform_2', & \swform_1' \sim_{\Gamma \cup \{\varphi_1, \neg \varphi_2\}} \swform_2'', \\ \swform_1'' \sim_{\Gamma \cup \{\neg \varphi_1, \varphi_2\}} \swform_2', \text{ and} & \swform_1'' \sim_{\Gamma \cup \{\neg \varphi_1, \neg \varphi_2\}} \swform_2''.\end{array}\]
\end{lemma}
\begin{proof}
  We show why the first equivalence is true; the remaining cases are similar.
  Let $\swform_1 = \varphi_1 \cond \swform_1' : \swform_1''$ and
  $\swform_2 = \varphi_2 \cond \swform_2' : \swform_2''$.
  If $(w,\sigma) \in \sat{\Gamma \cup \{\varphi_1, \varphi_2\}}$,
  then also $(w,\sigma) \in \sat{\Gamma}$, so
  \begin{align*}
  \sat{\swform_1'}(w,\sigma) &= \sat{\swform_1}(w,\sigma) = \sat{\swform_2}(w,\sigma) = \sat{\swform_2'}(w,\sigma).
  \qedhere
  \end{align*}
\end{proof}


\begin{theorem}\label{thm:completeness-swmso}
  For finite $\Gamma$ we have
  $\swform_1 \sim_\Gamma \swform_2$ if and only if $\Gamma \vdash \swform_1 \approx \swform_2$.
\end{theorem}
\begin{proof}
  Soundness can be proved by simply checking the 
  validity
  of each axiom.
  For completeness, note that if $\Gamma$ is inconsistent,
  then immediately $\Gamma \vdash \swform_1 \approx \swform_2$
  by Proposition~\ref{prop:swmso-theorems}(1).
  In the rest of the proof we may therefore assume that $\Gamma$ is consistent.
  
  
  The proof now proceeds by induction on 
  $|\swform_1|_?+|\swform_2|_?$, where $|\swform|_?$
  is 
  defined as follows.
  \[|\swform|_? = \begin{cases} 0, & \text{if } \swform = r \\ 1 + |\swform'|_? + |\swform''|_? ,& \text{if } \swform = \varphi \cond \swform' : \swform'' \end{cases}\]
  
  Case $|\swform_1|_? + |\swform_2|_? = 0$:
  In this case, $\swform_1 = r_1$ and $\swform_2 = r_2$ for some $r_1,r_2 \in R$.
  Since $r_1 = \sat{\swform_1}(w,\sigma) = \sat{\swform_2}(w,\sigma) = r_2$
  by assumption, we get $\Gamma \vdash \swform_1 \approx \swform_2$ by reflexivity.
  
  Case $|\swform_1|_? + |\swform_2|_? > 0$:
  In this case, without loss of generality, $\swform_1 = \varphi \cond \swform_1' : \swform_1''$.
  From the semantics, we have that $\swform_1 \sim_{\Gamma \cup \{\varphi \}} \swform_1'$ and $\swform_1 \sim_{\Gamma \cup \{\neg \varphi \}} \swform_1''$, so $\swform_2 \sim_{\Gamma \cup \{\varphi \}} \swform_1'$ and $\swform_2 \sim_{\Gamma \cup \{\neg \varphi \}} \swform_1''$.
  From the inductive hypothesis, we have that
  \(
  \Gamma \cup \{\varphi\} \vdash \swform_1' \approx \swform_2, \text{ and }
  \Gamma \cup \{\neg \varphi\} \vdash \swform_1'' \approx \swform_2
  .
\)
  From axiom (S4), we conclude that 
 \( \Gamma \vdash \swform_1 \approx \swform_2.\)
\end{proof}

\subsection{\cwMSO\ Without Sums}
We 
present
a complete axiomatization of a fragment of $\cwMSO$
in which $+$
is the only allowed sum operator.
Let $\cwMSO(?, +)$ be the fragment 
given by
\[\cwform ::= \nil \mid \textstyle{\prod_x} \swform \mid \varphi \cond \cwform_1 : \cwform_2 \mid \cwform_1 + \cwform_2,\]
where $\swform$ is a \swMSO\ formula and $\varphi$ a \MSO\ formula.
The corresponding first-order fragment \cwFO(?,+) is
obtained from the same grammar but letting $\swform$ be a \swFO\ formula
and $\varphi$ a \FO\ formula.
Droste and Gastin studied the first-order fragment \cwFO(?,+) 
in \cite{DG19},
where they showed that it is expressively equivalent to
aperiodic finitely ambiguous weighted automata.
This result contrasts the situation for
the full first-order \cwFO,
which they show to be expressively equivalent
to aperiodic polynomially ambiguous weighted automata.
Here \emph{aperiodic} means that there exists an integer $m \geq 1$
such that for any word $w$, $w$ concatenated with itself $m$ times is accepted
if and only if $w$ concatenated with itself $m+1$ times is accepted,
\emph{polynomially ambiguous} means that there is a polynomial $p$
such that each word $w$ has at most $p(|w|)$ successful runs,
and \emph{finitely ambiguous} means that the polynomial is constant.
We are not aware of a similar characterization
of the second-order fragment \cwMSO(?,+).
In \cite{GM18} it is shown that \emph{adding} various additional operators
to the logic does not increase its expressivity,
but the question of the expressive power
of various fragments of the logic is not addressed.
In the following we give some examples of the expressivity of the fragment \cwMSO(?,+).

\begin{example}
  Consider again Example \ref{ex:1}.
  The formula in that example does not belong to 
  \cwMSO(?,+),
  because it uses the general sum $\sum_x$.
  Instead we can count the number of $a$'s that appear before any $b$'s in a word.
  To do this, consider the formula
  $\varphi = P_a(x) \land \forall y.(P_b(y) \rightarrow x \leq y)$,
  and let $\Psi = \varphi \cond 1 : 0$ and $\Phi = \prod_x \Psi$.
  For the word $w = abaa$ we then get
  \begin{align*}
    \sat{\Phi}(w,\sigma) &= \sat{\Psi}(w,\sigma[x \mapsto 1])\sat{\Psi}(w,\sigma[x \mapsto 2]) \\
    &= \sat{\Psi}(w,\sigma[x \mapsto 3])\sat{\Psi}(w,\sigma[x \mapsto 4]) = \{1000\},
  \end{align*}
  which correctly tells us that there is one $a$ before any $b$'s.
  It is a simple matter to adapt this to also count
  the collective number of different things,
  such as the total number of $a$'s and $c$'s before any $b$'s.
  However, we can also, in some sense, count individually different things.
  If for example we want to count separately the number of $a$'s
  and the number of $b$'s in a word, we can let
  \begin{align*}
    \varphi_1 = P_a(x), \quad \Psi_1 = \varphi_1 \cond 1 : 0, \quad \Phi_1 = \prod_x \Psi_1, \\
    \varphi_2 = P_b(x), \quad \Psi_2 = \varphi_2 \cond 2 : 0, \quad \Phi_2 = \prod_x \Psi_2,
  \end{align*}
  and finally $\Phi = \Phi_1 + \Phi_2$.
  If we again take $w = abaa$, then
  \[\sat{\Phi}(w,\sigma) = \sat{\Phi_1}(w,\sigma) \munion \sat{\Phi_2}(w,\sigma) = \multiset{1011, 0200},\]
  and by counting off the number of $1$'s in this multiset,
  we obtain the number of $a$'s,
  and likewise the number of $2$'s gives number of $b$'s.
\end{example}

For a formula $\cwform$,
let $\var(\cwform)$ be the set of variables used in $\cwform$,
and let $\cwform[y/x]$ be the formula resulting from
replacing the variable $x$ with the variable $y$.
The axioms for the fragment $\cwMSO(?,+)$ are then given in Table \ref{tab:cwfo-axioms}.
Axioms ($C1$)-($C3$) give standard properties of sum,
whereas ($C4$) and ($C5$) take care of the product.
Axioms ($C6$)-($C9$) are similar to the axioms for $\swMSO$,
and finally, axiom ($C10$) shows how sum distributes over the conditional operator.

\begin{table}
  \centering
  \begin{tabular}{l l}
    \hline
    ($C1$): & $\Gamma \vdash \cwform + \nil \approx \cwform$ \\[1ex]
    ($C2$): & $\Gamma \vdash \cwform_1 + \cwform_2 \approx \cwform_2 + \cwform_1$ \\[1ex]
    ($C3$): & $\Gamma \vdash (\cwform_1 + \cwform_2) + \cwform_3 \approx \cwform_1 + (\cwform_2 + \cwform_3)$ \\[1ex]
    ($C4$): & $\begin{aligned}&\Gamma \vdash \swform_1 \approx \swform_2 \text{ implies } \Gamma \vdash \textstyle \prod_x \swform_1 \approx \prod_x \swform_2 \\[-1ex]
    &\text{if } x \text{ is not free in } \Gamma\end{aligned}$ \\[2ex]
    ($C5$): & $\Gamma \vdash \prod_x \swform \approx \prod_y \swform[y / x]$ if $y \notin \var(\swform)$ \\[1ex]
    ($C6$): & $\Gamma \vdash \cwform_1 \approx \cwform_2$ implies $\Gamma \cup \{\varphi\} \vdash \cwform_1 \approx \cwform_2$ \\[1ex]
    ($C7$): & $\Gamma \vdash \neg \varphi \cond \cwform_1 : \cwform_2 \approx \varphi \cond \cwform_2 : \cwform_1$ \\[1ex]
    ($C8$): & if $\Gamma \vdash \varphi$ then $\Gamma \vdash \varphi \cond \cwform_1 : \cwform_2 \approx \cwform_1$ \\[1ex]
    ($C9$): & $\begin{aligned}&\Gamma \cup \{\varphi\} \vdash \cwform_1 \approx \cwform \text{ and } \Gamma \cup \{\neg \varphi\} \vdash \cwform_2 \approx \cwform \\[-1ex]
    &\text{implies } \Gamma \vdash \varphi \cond \cwform_1 : \cwform_2 \approx \cwform \end{aligned}$ \\[2ex]
    ($C10$): & $\Gamma \vdash (\varphi \cond \cwform' : \cwform'') + \cwform \approx \varphi \cond (\cwform' + \cwform) : (\cwform'' + \cwform)$ \\
    \hline \\
  \end{tabular}
  \caption{Axioms for $\cwMSO(?,+)$.}
  \label{tab:cwfo-axioms}
\end{table}

Since all of the axioms for $\swMSO$ are also included
in the axiomatization for $\cwMSO$ (because both include the conditional operator),
we get that the theorems we derived in Proposition \ref{prop:swmso-theorems}
are also derivable for $\cwMSO(?,+)$.
Likewise, Lemma \ref{lem:completeness-swmso} also carries over to $\cwMSO(?,+)$

Our first lemma shows the connection between the product operator $\prod_x$
and the first-order 
universal
quantifier $\forall x$,
which implies that
axiom ($C4$) is sound.

\begin{lemma}\label{lem:forall}
  If $x$ does not appear as a free variable in $\Gamma,$ then 
  $\prod_x \swform_1 \sim_\Gamma \prod_x \swform_2$ if and only if $\swform_1 \sim_{\Gamma} \swform_2$.
\end{lemma}
\begin{proof}
  (${\implies}$) $\prod_x \swform_1 \sim_\Gamma \prod_x \swform_2$ implies that
  $\sat{\swform_1}(w, \sigma[x \mapsto i]) = \sat{\swform_2}(w, \sigma[x \mapsto i])$
  for all $i$ and $(w,\sigma)$ such that $(w, \sigma) \models \Gamma$.
  In particular, 
  $\sat{\swform_1}(w,\sigma) = \sat{\swform_2}(w,\sigma)$
  for all $(w,\sigma) \models \Gamma$,
  so $\swform_1 \sim_{\Gamma} \swform_2$.
  
  ($\impliedby$) $\swform_1 \sim_{\Gamma} \swform_2$
  means that $\sat{\swform_1}(w,\sigma) = \sat{\swform_2}(w,\sigma)$
  for all $(w,\sigma) \models \Gamma$.
  This implies that $\sat{\swform_1}(w,\sigma[x \mapsto i]) = \sat{\swform_2}(w,\sigma[x \mapsto i])$
  for all $i$ and $(w,\sigma[x \mapsto i]) \models \Gamma$.
  But since $x$ does not appear free in $\Gamma$, 
  $(w,\sigma[x \mapsto i]) \models \Gamma$ if and only if $(w,\sigma) \models \Gamma$, and therefore
  $\sat{\swform_1}(w,\sigma[x \mapsto i]) = \sat{\swform_2}(w,\sigma[x \mapsto i])$
  for all $i$ and $(w,\sigma) \models \Gamma$.
  This in turn implies 
  $\sat{\prod_x \swform_1}(w,\sigma) = \sat{\prod_x \swform_2}(w,\sigma)$
  for all $(w,\sigma) \models \Gamma$, so $\prod_x \swform_1 \sim_\Gamma \prod_x \swform_2$.
\end{proof}

%

A key part of the proof of completeness is to put formulas
into the following notion of normal form,
where occurrences of the conditional operator are grouped together
and all come before any sum or product is applied.

\begin{definition}
  A $\cwMSO(?, +)$ formula $\cwform$ is in \emph{normal form} if 
  $\cwform$ is generated by the following grammar:
\[ N ::= \varphi \cond N_1 : N_2 \mid M \mid \nil \quad \text{and} \quad 
 M ::= {\textstyle\prod_x} \swform \mid M_1 + M_2.\]
\end{definition}

Every
$\cwMSO(?,+)$ has an equivalent normal form,
which will allow us to only reason about formulas in normal form in the proof.
In order to show this, we make use of the following technical lemma,
which takes care of the case of the sum operator.

\begin{lemma}\label{lem:conditional}
  If $\cwform_1$ and $\cwform_2$ are in normal form,
  then there exists a formula $\cwform$, also in normal form,
  such that $\Gamma \vdash \cwform \approx \cwform_1 + \cwform_2$.
\end{lemma}
\begin{proof}
  The proof is by induction on the maximum number of nested occurrences
  of the conditional operator within $\cwform_1$ and $\cwform_2$.
  Note that since these are in normal form,
  occurrences of the conditional operator will always appear
  consecutively as the outermost operators.
  Formally, we define, on formulas $\cwform$ in normal form,
  the following function which counts the number of nested occurrences of the conditional operator:
  \[\ncond(\cwform) = \begin{cases} 1 + \max\{\ncond(\cwform'),\ncond(\cwform'')\} & \text{if } \cwform = \varphi \cond \cwform' : \cwform'' \\ 0 & \text{otherwise.}\end{cases}\]
  Let $k = \max\{\ncond(\cwform_1),\ncond(\cwform_2)\}$.
  
  $k = 0$:
  This case follows essentially from ($C1$).
  
  $k > 0$:
  We have three cases to consider:
  (1) $\ncond(\cwform_1) = \ncond(\cwform_2)$,
  (2) $\ncond(\cwform_1) < \ncond(\cwform_2)$, or
  (3) $\ncond(\cwform_1) > \ncond(\cwform_2)$.

  (1) Consider $\cwform_1 = \varphi_1 \cond \cwform_1' : \cwform_1''$
  and $\cwform_2 = \varphi_2 \cond \cwform_2' : \cwform_2''$.
  Now, by three applications of axiom ($C10$),
  we get
  \begin{align*}
    \Gamma &\vdash \varphi_1 \cond \cwform_1' : \cwform_1'' + \varphi_2 \cond \cwform_2' : \cwform_2'' \\
    &\phantom{{}\vdash{}}\approx \varphi_1 \cond (\varphi_2 \cond \cwform_1' + \cwform_2' : \cwform_1' + \cwform_2'') \\
    &\phantom{{}\approx{} \varphi_1}: (\varphi_2 \cond \cwform_1'' + \cwform_2' : \cwform_1'' + \cwform_2'').
  \end{align*}
  Since $\ncond(\cwform_1' + \cwform_2') < k$, $\ncond(\cwform_1' + \cwform_2'') < k$,
  $\ncond(\cwform_1'' + \cwform_2') < k$, and $\ncond(\cwform_1'' + \cwform_2'') < k$,
  the induction hypothesis gives formulas $\cwform'$, $\cwform''$, $\cwform'''$, and $\cwform''''$,
  all in normal form, such that
  $\Gamma \vdash \cwform' \approx \cwform_1' + \cwform_2'$,
  $\Gamma \vdash \cwform'' \approx \cwform_1' + \cwform_2''$,
  $\Gamma \vdash \cwform''' \approx \cwform_1'' + \cwform_2'$, and
  $\Gamma \vdash \cwform'''' \approx \cwform_1'' + \cwform_2''$.
  Thus
  \[\cwform = \varphi_1 \cond (\varphi_2 \cond \cwform' : \cwform'') : (\varphi_2 \cond \cwform''' : \cwform'''')\]
  is in normal form and satisfies $\Gamma \vdash \cwform \approx \cwform_1 + \cwform_2$.
  
  (2), (3) These cases are simpler versions of case (1).
\end{proof}

\begin{lemma}\label{lem:normal-form}
  For each $\Gamma$ and $\cwMSO(?, +)$ formula $\cwform$,
  there is a formula $\cwform'$ in normal form such that
  $\Gamma \vdash \cwform \approx \cwform'$.
\end{lemma}
\begin{proof}
  The proof is by induction on the structure of $\cwform$.
  
  $\cwform = \nil$ or $\cwform = \prod_x \swform$:
  Then, 
  $\cwform$ is already in normal form.
  
  $\cwform = \cwform_1 + \cwform_2$:
  By induction hypothesis, there exist formulas $\cwform_1'$ and $\cwform_2'$,
  both in normal form, such that $\Gamma \vdash \cwform_1 \approx \cwform_1'$
  and $\Gamma \vdash \cwform_2 \approx \cwform_2'$.
  By Lemma \ref{lem:conditional},
  there exists a formula $\cwform'$ in normal form such that
  $\Gamma \vdash \cwform' \approx \cwform_1' + \cwform_2'$.
  By congruence we get $\Gamma \vdash \cwform_1 + \cwform_2 \approx \cwform_1' + \cwform_2'$,
  so $\Gamma \vdash \cwform \approx \cwform'$.
  
  $\cwform = \varphi \cond \cwform_1 : \cwform_2$:
  By induction hypothesis there exist $\cwform_1'$ and $\cwform_2'$
  in normal form such that $\Gamma \vdash \cwform_1 \approx \cwform_1'$
  and $\Gamma \vdash \cwform_2 \approx \cwform_2'$.
  Then $\cwform' = \varphi \cond \cwform_1' : \cwform_2'$ is in normal form
  and, by congruence, $\Gamma \vdash \cwform \approx \cwform'$.
\end{proof}

Notice that for formulas in normal form,
if it is not the case that $\cwform = \varphi \cond \cwform_1 : \cwform_2$,
then $\cwform$ can not contain any conditional statements at all,
and hence $\cwform$ must be of the form $\cwform = \sum_{i = 1}^k \prod_x \cwform_i$ for some $k$ (axioms (C2) and (C3) allow us to use this finite sum notation).
The following series of lemmas shows that for formulas of this form, it is enough
to consider each of the summands pairwise.

\begin{lemma}\label{lem:prod_formula}
  Given two formulas $\swform_1$ and $\swform_2$,
  there exists a formula $\varphi_{\swform_1,\swform_2}$ such that
  $(w,\sigma) \models \forall x.\varphi_{\swform_1,\swform_2}$ if and only if
  $\sat{\prod_x \swform_1}(w,\sigma) = \sat{\prod_x \swform_2}(w,\sigma)$.
  In particular,
  \[\prod_x \swform_1 \sim_{\Gamma \cup \{\forall x.\varphi_{\swform_1,\swform_2}\}} \prod_x \swform_2.\]
\end{lemma}
\begin{proof}
  Consider the sets $R_1$ and $R_2$ of values that appear in $\swform_1$ and $\swform_2$, respectively.
  If these sets are disjoint,
  then $\sat{\prod_x \swform_1}(w,\sigma) \neq \sat{\prod_x \swform_2}(w,\sigma)$
  for all $(w,\sigma)$, so we can take $\varphi_{\swform_1,\swform_2} = \false$.
  
  If they are not disjoint, consider any $r \in R_1 \cap R_2$.
  From Lemma \ref{lem:WtoMSO}, 
  $(w,\sigma) \models \varphi(\swform_1,r)$ if and only if
  $\sat{\swform_1}(w,\sigma) = r$; and 
  $(w,\sigma) \models \varphi(\swform_2,r)$
  if and only if $\sat{\swform_2}(w,\sigma) = r$.
  Now we take $\varphi^r_{\swform_1,\swform_2} = \varphi(\swform_1,r) \land \varphi(\swform_2,r)$ and 
  \[\varphi_{\swform_1,\swform_2} = \bigvee_{r \in R_1 \cap R_2} \varphi^r_{\swform_1,\swform_2}.\]
  We now have a formula $\varphi_{\swform_1,\swform_2}$ such that, for all $(w,\sigma)$,
  $(w,\sigma) \models \varphi_{\swform_1,\swform_2}$ if and only if
  $\sat{\swform_1}(w,\sigma) = \sat{\swform_2}(w,\sigma)$.
  This is equivalent to
  \begin{align*}
  \forall (w,\sigma).\forall i &\in \{1,\dots,|w|\}.(w,\sigma[x \mapsto i]) \models \varphi_{\swform_1,\swform_2} \\ \text{ iff } &\sat{\swform_1}(w,\sigma[x \mapsto i]) = \sat{\swform_2}(w,\sigma[x \mapsto i])
  \end{align*}
  which implies that
  $(w,\sigma) \models \forall x. \varphi_{\swform_1,\swform_2}$  iff $\sat{{\textstyle\prod_x} \swform_1}(w,\sigma) = \sat{{\textstyle \prod_x} \swform_2}(w,\sigma)$ for all $(w,\sigma)$.
\end{proof}

\begin{lemma}\label{lem:deduction}
  If $\Gamma \vdash \bigvee_{m=1}^n \varphi_m$ and for every $m$,
  it holds that $\Gamma \cup \{\varphi_m\} \vdash \cwform_1 \approx \cwform_2$,
  then $\Gamma \vdash \cwform_1 \approx \cwform_2$.
\end{lemma}
\begin{proof}
  The proof is by induction on $n$.
  The case of $n = 1$ is trivial: 
  we have assumed that $\Gamma \cup \{\varphi_1\} \vdash \cwform_1 \approx \cwform_2$,
  so $\Gamma \vdash \cwform_1 \approx \cwform_2$ from Proposition \ref{prop:cwmso-theorems}(2).
  Now, let $n = k+1$.
  We have $\Gamma \cup \{ \varphi_n \} \vdash \cwform_1 \approx \cwform_2$ and 
  $\Gamma \cup \{ \neg \varphi_n \} \vdash 
  \bigvee_{m=1}^k \varphi_m$, so by the inductive hypothesis, 
  $\Gamma \cup \{ \neg \varphi_n \} \vdash 
  \cwform_1 \approx \cwform_2$.
  Hence, Proposition \ref{prop:cwmso-theorems}(8) gives 
  $\Gamma \vdash \cwform_1 \approx \cwform_2$.
\end{proof}

\begin{lemma}\label{lem:completeness-sum}
  Let $\Gamma$ be finite.
  Assume $\cwform_1 = \sum_{i=1}^k \prod_x \swform_i$
  and $\cwform_2 = \sum_{j=1}^k \prod_x \swform_j'$
  with $\cwform_1 \sim_\Gamma \cwform_2$,
  and assume that for all $i$ and $j$
  $\prod_x \swform_i \sim_\Gamma \prod_x \swform_j'$ implies
  $\Gamma \vdash \prod_x \swform_i \approx \prod_x \swform_j'$.
  Then $\Gamma \vdash \cwform_1 \approx \cwform_2$.
\end{lemma}
\begin{proof}
  By definition, $\cwform_1 \sim_\Gamma \cwform_2$ means that
  for all $(w,\sigma) \in \sat{\Gamma}$
  there exists a permutation $(j_1, \dots, j_k)$ of $\{1,2,\ldots,k\}$, such that 
  for all $i$,
  \begin{equation}\label{eq:proof1}
  \sat{\textstyle{\prod_x} \swform_i}(w,\sigma) = \sat{\textstyle{\prod_x} \swform_{j_i}'}(w,\sigma)
  .
  \end{equation}
  
  By Lemma \ref{lem:prod_formula},
  for each such permutation $P = (j_1, \dots, j_k)$ there exist
  formulas $\varphi_{1,j_1}, \dots, \varphi_{k, j_k}$ such that
  \[\prod_x \swform_i \sim_{\Gamma \cup \{\forall x. \varphi_{i,j_i}\}} \prod_x \swform_{j_i}',\]
  and by assumption, this gives
  \begin{equation}\label{eq:proof2}
  \Gamma \cup \{\forall x. \varphi_{i, j_i}\} \vdash \prod_x \swform_i \approx \prod_x \swform_{j_i}'.
  \end{equation}
  
  For each permutation $P = \{j_1, \dots, j_k\}$, let
  \[\varphi_P = (\forall x.\varphi_{1,j_1}) \land \dots \land (\forall x. \varphi_{k, j_k}).\]
  By Equation \eqref{eq:proof1} and Lemma \ref{lem:prod_formula},
  for every $(w,\sigma) \in \sat{\Gamma}$ there exists a permutation
  $P = (j_1, \dots, j_k)$ such that we have $(w,\sigma) \models \varphi_P$.
  This means that for all $(w,\sigma) \in \sat{\Gamma}$
  we have $(w,\sigma) \models \bigvee_P \varphi_P$.
  By Corollary \ref{cor:completeness-mso},
  this means that $\Gamma \vdash \bigvee_P \varphi_P$.
  
  Now, from Equation \eqref{eq:proof2},
  we can use ($C6$) to get
  \[\Gamma \cup \{\varphi_P\} \cup \{\forall x.\varphi_{i,j_i}\} \vdash \prod_x \swform_i \approx \prod_x \swform'_{j_i},\]
  and together with $\Gamma \cup \{\varphi_P\} \vdash \forall x.\varphi_{i,j_i}$,
  this gives
  $\Gamma \cup \{\varphi_P\} \vdash \prod_x \swform_i \approx \prod_x \swform_{j_i}'$
  by Proposition \ref{prop:cwmso-theorems}(2).
  We can then use congruence to get
  \begin{equation}\label{eq:proof3}
  \Gamma \cup \{\varphi_P\} \vdash \sum_{i = 1}^k \prod_x \swform_i \approx \sum_{j = 1}^k \prod_x \swform_{j_i}'.
  \end{equation}
  Since $\sum_{j = 1}^k {\textstyle\prod_x} \swform_{j_i}'$
  is a permutation of $\Phi_2$, we get by axioms ($C2$) and ($C3$) that
  $\Gamma \cup \{\varphi_P\} \vdash \Phi_2 \approx \sum_{j = 1}^k {\textstyle\prod_x} \swform_{j_i}'$, so
  \begin{equation}\label{eq:proof4}
  \Gamma \cup \{\varphi_P\} \vdash \Phi_1 \approx \Phi_2
  \end{equation}
  by Equation \eqref{eq:proof3}.
  By Lemma \ref{lem:deduction}, Equation \eqref{eq:proof4}
  together with the fact that $\Gamma \vdash \bigvee_P \varphi_P$ gives
  $
  \Gamma \vdash \cwform_1 \approx \cwform_2. 
  $
\end{proof}

We can now prove completeness for formulas in normal form,
and by Lemma \ref{lem:normal-form},
this extends to all formulas.

\begin{lemma}\label{lem:nf-completeness}
  If $\cwform_1$ and $\cwform_2$ are in normal form and $\Gamma$ is finite,
  then $\cwform_1 \sim_\Gamma \cwform_2$ implies $\Gamma \vdash \cwform_1 \approx \cwform_2$.
\end{lemma}
\begin{proof}
  By Proposition \ref{prop:swmso-theorems}(1), we may assume that $\Gamma$ is consistent.
  We note that for a formula $\cwform$ in normal form, if $\cwform$ is not a conditional and 
  $\cwform \neq \nil$, then for every $(w,\sigma)$, $|\sat{\cwform}(w,\sigma)|>0$, and therefore $\cwform \not \sim_\Gamma \nil$.
   The proof now proceeds by induction on $d = 
   \depth(\cwform_1)
   +
   \depth(\cwform_2)$,
  where
  $\depth(\cwform) = 0$ if $\cwform = \nil$ or $\cwform = \prod_x \swform$ and
  $\depth(\cwform) = 1 + \max\{\depth(\cwform'),\depth(\cwform'')\}$ if $\cwform = \varphi \cond \cwform' : \cwform''$ or $\cwform = \cwform' + \cwform''$.
  
  \emph{Case $d=0$.}
  We distinguish the following two subcases. 
  
  At least one of $\cwform_1$ and $\cwform_2$ is $\nil$: Then from the observation above, $\cwform_1 = \nil = \cwform_2$, and therefore $\Gamma \vdash \nil \approx \nil$ by reflexivity.
  
  $\cwform_1 = \prod_{x_1} \swform_1$ and $\cwform_2 = \prod_{x_2} \swform_2$:
 	In this case, we can find some $x \notin \var(\swform_1) \cup \var(\swform_2)$ that does not appear in $\Gamma$.
 	Then, $\prod_x \swform_1[x / x_1] \sim_\Gamma \prod_x \swform_2[x / x_2]$.
 	By Lemma \ref{lem:forall} we get $\swform_1[x / x_1] \sim_{\Gamma} \swform_2[x / x_2]$,
 	and by completeness of $\swMSO$, this implies
 	$\Gamma \vdash \swform_1[x / x_1] \approx \swform_2[x / x_2]$.
 	We can then use axiom ($C4$) to obtain
 	$\Gamma \vdash \prod_x \swform_1[x / x_1] \approx \prod_x \swform_2[x / x_2]$,
 	and finally use axiom ($C5$) to obtain
 	$\Gamma \vdash \prod_{x_1} \swform_1 \approx \prod_{x_2} \swform_2$.
  
  \emph{Case $d>0$.}
  We distinguish the following two subcases. 
  
  Without loss of generality, $\cwform_1 = \varphi \cond \cwform_1' : \cwform_1''$:
 	Then, from $\cwform_1 \sim_\Gamma \cwform_2$ we get 
 	$\cwform_1' \sim_{\Gamma \cup \{\varphi\}} \cwform_2$ and 
 	$\cwform_1'' \sim_{\Gamma \cup \{\neg \varphi\}} \cwform_2$, and by the inductive hypothesis this yields ${\Gamma \cup \{\varphi\}} \vdash \cwform_1' \approx \cwform_2$ and ${\Gamma \cup \{\neg \varphi\}} \vdash \cwform_1'' \approx \cwform_2$.
 	Axiom ($C9$) then gives us that 
 	$\Gamma \vdash \cwform_1 \approx \cwform_2$.
   
  $\cwform_1 = \sum_{i=1}^k \prod_x \swform_i$
  and $\cwform_2 = \sum_{j=1}^{k'} \prod_x \swform_j'$: 
 	Then, we must have $k = k'$
 	since otherwise $|\sat{\cwform_1}(w,\sigma)| \neq |\sat{\cwform_2}(w,\sigma)|$,
 	contradicting $\cwform_1 \sim_\Gamma \cwform_2$.
 	By the induction hypothesis, $\prod_x \swform_i \sim_\Gamma \prod_x \swform_j'$
 	implies $\Gamma \vdash \prod_x \swform_i \approx \prod_x \swform_j'$,
 	so Lemma~\ref{lem:completeness-sum} yields
 	$\Gamma \vdash \cwform_1 \approx \cwform_2$.
\end{proof}

\begin{theorem}[Completeness for $\cwMSO(\cond,+)$]\label{thm:completeness-cwmso}
  For every finite $\Gamma$ and $\cwMSO(\cond,+)$ formulas $\cwform_1$ and $\cwform_1$, we have
  $\cwform_1 \sim_\Gamma \cwform_2 \text{ if and only if } \Gamma \vdash \cwform_1 \approx \cwform_2$.
\end{theorem}
\begin{proof}
	We prove only completeness. 
  Assume $\cwform_1 \sim_\Gamma \cwform_2$.
  By Lemma \ref{lem:normal-form}, there exist formulas $\cwform_1'$ and $\cwform_2'$,
  both in normal form, such that $\Gamma \vdash \cwform_1 \approx \cwform_1'$
  and $\Gamma \vdash \cwform_2 \approx \cwform_2'$.
  By soundness,
  this implies $\cwform_1 \sim_\Gamma \cwform_1'$ and $\cwform_2 \sim_\Gamma \cwform_2'$,
  so $\cwform_1' \sim_\Gamma \cwform_2'$.
  Since these are in normal form, Lemma \ref{lem:nf-completeness}
  gives $\Gamma \vdash \cwform_1' \approx \cwform_2'$,
  and we conclude $\Gamma \vdash \cwform_1 \approx \cwform_2$.
\end{proof}

\section{An Axiomatization for Full \cwMSO}


In this section, we give a complete axiomatization for the full \cwMSO.
First, we prove a result about weighted automata that will help with the completeness proof and that will help explain one of the axioms.
We follow the definition of weighted automata, using their abstract semantics, from \cite{GM18}.

\subsection{Weighted Automata}

An $R$-weighted automaton over $\Sigma$ is a quintuple $A = (Q,\Delta,\wgt,I,F)$, where $Q$ is a nonempty and finite set of states, $I, F \subseteq Q$ are, respectively, the initial and final states of the automaton, $\Delta \subseteq Q \times \Sigma \times Q$ is the transition relation, and $\wgt : \Delta \to R$ assigns a weight from $R$ to each transition of the automaton.
A run of $A$ on a word $w \in \Sigma^*$ of length $n$ is a sequence $\delta_1
\delta_2\cdots \delta_n \in \Delta^n$, where for every $i \leq n$, $\delta_i = (q_i,a_i,q_{i+1})$, and $w = a_1 a_2 \cdots a_n$.
It is an accepting run if $q_1 \in I$ and $q_{n+1} \in F$.
We can extend the weight function $\wgt$ on runs, such that $\wgt(\delta_1
\delta_2\cdots \delta_n) = \wgt(\delta_1)\wgt(\delta_2)\cdots \wgt(\delta_n)$.
We denote as $\rho(A,w)$ the set of runs of $A$ on $w$.

The semantics of $A = (Q,\Delta,\wgt,I,F)$ is given by a function $\sat{\cdot}:\Sigma^+ \to \multi{R^*}$ in the following way:
\begin{align*}
	\sat{A}(w) = \multiset{\wgt(\rho) \mid \rho \text{ is an accepting run of $A$ on $w$}}.
\end{align*}

\begin{theorem}[\cite{GM18}]\label{thm:form-to-A}
	For every closed $\cwMSO$ formula $\cwform$, there is an $R$-weighted automaton over $\Sigma$, $A$, such that for every $w \in+\Sigma^*$, 
	$\sat{\cwform}(w) = \sat{A}(w)$.
\end{theorem}

\begin{remark}
	Theorem \ref{thm:form-to-A} applies only to closed formulas, yet we mainly work with possibly open formulas. 
	But this is not really a limitation, as every 
	formula $\cwform$ with a set $V$ of free variables can be thought of as a closed formula over the extended alphabet $\Sigma \cup V$.
\end{remark}

We  extend the semantic equivalence $\sim$ of formulas to weighted automata as expected, but we also introduce a bounded version of this equivalence. Specifically, for every $n \geq 0$, and for every pair $A_1,A_2$ of automata, $A_1 \sim_n A_1$, if for every $w \in \Sigma^+$ of length at most $n$, $\sat{A_1}(w) = \sat{A_2}(w)$.

\begin{theorem}\label{thm:finite-equivalence-A}
	Let $A_1$ and $A_2$ be two $R$-weighted automata over $\Sigma$, such that $A_1$ has $n_1$ states and $A_2$ has $n_2$ states.
	Then, $A_1 \sim_{n_1 + n_2 - 1} A_2$ if and only if $A_1 \sim A_2$.
\end{theorem}

\begin{proof}
	The ``if'' direction of the theorem is trivial, and therefore we prove the ``only if'' direction.
	Let  $n = n_1 + n_2$, and let $A_1 = (Q_1,\Delta_1,\wgt,I_1,F_1)$ and $A_2 = (Q_2,\Delta_2,\wgt,I_2,F_2)$ --- the weight function is considered the same for the two automata, for convenience.
	For every word $wa \in \Sigma^+$, $\gamma r \in R^{|wa|}$, $i = 1,2$, and $S$ a set or multiset of states from $Q_i$, we define 
	$Q_i(S,\varepsilon,\varepsilon)= S$, and
	\begin{align*}
		Q_i(S,wa,\gamma r) =  
	\multiset{ &q \in Q_i \mid  \exists q' {\in} Q_i(S,w,\gamma) \text{ such that } \\  & (q',a,q) {\in} \Delta_i 
		\text{ and }  \wgt((q',a,q)) {=} r  }.
	\end{align*}
Let $Q(S,w,\gamma) = Q_1(S\cap Q_1,w,\gamma) \munion Q_2(S\cap Q_2,w,\gamma)$ and let $I = I_1 \cup I_2$.
	
	We assume
	that  $A_1 \sim_{n - 1} A_2$
	and we use strong induction on $|w|$ to prove that for every string $w$, 
	$\sat{A_1}(w) = \sat{A_2}(w)$.
	The cases for $|w| < n$ are immediate from our assumptions.
	We now consider the case where
	$w$ is of length $m > n-1$, and for every word $w'$ of length less than $m$, $\sat{A_1}(w') = \sat{A_2}(w')$. 
	Let $\rho$ be a sequence of transitions from $A_1$ or $A_2$ of length $m$.
	We prove that $\wgt(\rho)$ appears the same number of times in $\sat{A_1}(w)$ and in $\sat{A_2}(w)$, which suffices to complete the inductive proof.
	
	Since $m \geq n$, $w$ and $\rho$ have at least $n+1$ prefixes each, say $w_i$ and $\rho_i$ of length $i$, where $0 \leq i \leq n$.
	We can fix an ordering of the states of the two automata, and therefore,
	for each $i$, we can think of $Q(I,w_i,\wgt(\rho_i))$ as a vector of nonnegative integers of dimension $n$.
	These are at least $n + 1$ vectors of dimension at most $n$, so they must be linearly dependent.
	Therefore, there is some $0 < i_0 \leq n$, such that $Q(I,w_{i_0},\wgt(\rho_{i_0}))$ is a linear combination of $\{ Q(I,w_i,\wgt(\rho_i)) \mid 0 \leq i < i_0 \}$ (with rational coefficients), which we denote as $$Q(I,w_{i_0},\wgt(\rho_{i_0})) = \lambda((Q(I,w_i,\wgt(\rho_i)))_{i=0}^{i_0 - 1}).$$
	Let $w = w_{i_0} w'$ and $\rho = \rho_{i_0} \rho'$.
	By a direct inductive argument, 
	\begin{align}
	Q(I,w,\wgt(\rho)) = \lambda((Q(I,w_iw',\wgt(\rho_i\rho')))_{i=0}^{i_0 - 1}). \label{eq:linear-comb}
	\end{align}
	We observe that the number of times that
	$\wgt(\rho)$ appears in $\sat{A_1}(w)$ and in $\sat{A_2}(w)$ is the cardinality of $Q(I,w,\wgt(\rho)) \cap F_1$ and of $Q(I,w,\wgt(\rho)) \cap F_2$, respectively.
	Therefore,  for  $ k = 1,2 $,
	\begin{align*}
		\sat{A_k}(w)(\wgt(\rho)) =~& |Q(I,w,\wgt(\rho)) \cap F_k| \\
		=~& 
		| \lambda((Q(I,w_iw',\wgt(\rho_i\rho')))_{i=0}^{i_0 - 1}) \cap F_k |
		\tag*{from \eqref{eq:linear-comb}}
		\\
		=~& 
		 \lambda((|Q(I,w_iw',\wgt(\rho_i\rho')) \cap F_k|)_{i=0}^{i_0 - 1}) \\
		=~&
		 \lambda(( |\sat{A_k}(w_iw')(\wgt(\rho_i\rho'))| )_{i=0}^{i_0 - 1}), 
	\end{align*}
but, from the inductive hypothesis, for $i {=} 0$ to $i_0{-}1$, 
$
|\sat{A_1}(w_iw')(\wgt(\rho_i\rho'))| 
=
|\sat{A_2}(w_iw')(\wgt(\rho_i\rho'))| 
$, 
and therefore
			$
				\sat{A_1}(w)(\wgt(\rho))
		=
		\sat{A_2}(w)(\wgt(\rho)), 
		$
	which completes the proof.
%
\end{proof}

\begin{corollary}\label{cor:automataP}
	The equivalence problem for weighted automata is in 
	\P.
\end{corollary}

\begin{proof}
	We observe from the proof of Theorem \ref{thm:finite-equivalence-A} that for every automaton $A$, word $w$, and $\gamma \in R^{|w|}$, that $Q(I,w,\gamma)$ can be computed iteratively in polynomial time, with respect to $|A|$ and $|w|$.
	Furthermore, we observe that two weighted automata $A_1 = (Q_1,\Delta_1,\wgt,I_1,F_1)$ and $A_2 = (Q_2,\Delta_2,\wgt,I_2,F_2)$  are not equivalent, if, and only if, there is a string $w$ of length at most $|Q_1|+|Q_2|$, and a $\gamma \in R^{|w|}$, such that
	\[
	|Q(I,w,\gamma) \cap F_1| ~~\neq~~ |Q(I,w,\gamma) \cap F_2|.
	\]
	We now show how
	to try all possible $w$ and $\gamma$ of length at most $|Q_1|+|Q_2|$, 
	in polynomial time. 
	Let $R_A \subseteq R$ be the set of weights that appear in $A_1$ or $A_2$. 
	Starting from $\Lambda := \{Q(I,\varepsilon,\varepsilon)\}$, repeat the following  $|Q_1|+|Q_2|$ times:
	\begin{itemize}
		\item compute $\Lambda := \{ Q(I,w a,\gamma r) \mid Q(I,w,\gamma) \in \Lambda, \ r \in R_A, a \in \Sigma\}$;  and 
		\item  replace $\Lambda$ by a maximal subset of linearly independent elements of $\Lambda$. 
	\end{itemize}
	If at any step, for some element $e \in \Lambda$, $	|e \cap F_1| ~~\neq~~ |e \cap F_2|$, then we reject; otherwise we accept.
	We maintain at most $|Q_1|+|Q_2|$ values in $\Lambda$ at every step, and 
	both steps can be done in polynomial time.
	Therefore, this is a polynomial-time algorithm for the equivalence problem.
\end{proof}

\begin{remark}
	There are similarities between the proof of our complexity bound (Theorem \ref{thm:sat-is-undec} and Corollary \ref{cor:automataP}) and \cite{S61} and \cite{cortes2007lp}. However, the techniques in these papers do not directly apply in our case, due to the nature of abstract semantics,  which maintain the information of all the runs of the automata.
	Furthermore, we observe that, using the techniques from \cite{KieferMQWW2013}, one can possibly further improve on the complexity bound of Corollary \ref{cor:automataP}.
\end{remark}

\begin{corollary}\label{cor:finite-sat-Form}
	There is a computable function $\ell : \nat \to \nat$, such that for every pair $\cwform_1$ and $\cwform_2$ of $\cwMSO$ formulas, and environment $\Gamma$, $\cwform_1 {\sim_\Gamma} \cwform_2$ if and only if 
	for every $(w,\sigma) {\in} \sat{\Gamma}$ of length at most $\ell(|\cwform_1|{+}|\cwform_2| {+} |\Gamma|)$, $\sat{\cwform_1}(w,\sigma) = \sat{\cwform_2}(w,\sigma)$.
\end{corollary}

\begin{proof}
	The corollary results from Theorems \ref{thm:form-to-A} and \ref{thm:finite-equivalence-A}, and the observations that $\cwform_1 \sim_\Gamma \cwform_2$ 
	iff 
	$\bigwedge \Gamma \cond \cwform_1 : \nil \sim \bigwedge \Gamma \cond \cwform_2 : \nil$, and that the proof of Theorem \ref{thm:form-to-A} in \cite{GM18} is constructive.
\end{proof}

\begin{corollary}
	The equational validity problem for \cwMSO\ is decidable.
\end{corollary}

\subsection{An Axiomatization of full \cwMSO}

We now present the full axiomatization for \cwMSO.
For brevity, we only use the second-order version of the sum operator
and elide the first-order versions of these axioms.
Specifically, in the following, axioms (C11) to (C16) have straightforward first-order versions that are omitted, and (C17) and the upcoming formula 
$\cwform_1 \leq_l \cwform_2$ can be rewritten to accommodate mixed sequences of both first- and second-order variables; the soundness and completeness proofs then go through in a straightforward way.
We use the notation $\vec{X}$ for a sequence of variables, $X_1,X_2,\ldots,X_k$, and $|\vec{X}| = k$.
This notation can be extended to the sum operator, such that $\sum_{\vec{X}}$ denotes $\sum_{X_1}\cdots \sum_{X_k}$.
For $|\vec{X}| = |\vec{Y}|$, we use $\vec{X} \neq \vec{Y}$ for $$\exists x. \bigvee_{i=1}^k (X_i(x) \land \neg Y_i(x)) \lor (\neg X_i(x) \land Y_i(x)).$$
Let $\cwform_1= \sum_{\vec{X}} \varphi_1 \cond \prod_x \swform_1 : \nil$ and 
$\cwform_2= \sum_{\vec{X}} \varphi_2 \cond \prod_x \swform_2 : \nil$.
For every $l \geq 0$, we use $\cwform_1 \leq_l \cwform_2$ for
\newcommand{\ijnot}{\bigwedge_{ i \neq j }}
\begin{align*}
	&
	\bigwedge_{m=1}^l
	\forall \vec{X}^1 
	\vec{X}^2 
	\cdots \vec{X}^m.~ \exists \vec{Y}^1 
	\vec{Y}^2 
	\cdots \vec{Y}^m.\hfill  \\
	&\left[\!\!\!\!
		\begin{array}{c} \displaystyle
			\left(
				\ijnot \vec{X}^i \neq \vec{X}^j \land \varphi_1(\vec{X}^i) 
				\rightarrow
				\ijnot \vec{Y}^i \neq \vec{Y}^j \land \varphi_2(\vec{Y}^i)
			\right)	
			\\[3ex]
			{\mathlarger{\mathlarger{\land} } } 
			\\[2ex]
			\displaystyle 
			\left[\!\!\!\!
			\begin{array}{c} \displaystyle
					\ijnot 
					\!\left(\!
					\varphi_1(\vec{X}^i)  
					\land 
					\forall x. {\bigwedge_{
						r {\in} \swform_1 
				 	}} 
					\varphi(\swform_1(\vec{X}^i),r) {\leftrightarrow} \varphi(\swform_1(\vec{X}^j),r)
					\!\right)\!
					\\[4ex]
			{\mathlarger{\mathlarger{\rightarrow}}}
					\\[2ex]
					\displaystyle
					\left[\!\!\!
					\begin{array}{c} \displaystyle
						\ijnot 
						\!\left(\!
						\varphi_2(\vec{Y}^i)  
						\land 
						\forall x. {\bigwedge_{
							 r {\in}\swform_2
						}} 
						\varphi(\swform_2(\vec{Y}^i),r) {\leftrightarrow} \varphi(\swform_2(\vec{Y}^j) ,r)
						\!\right)\!
						\\[4ex]
						{\mathlarger{\mathlarger{\land}}}
						\\[2ex]
						\displaystyle
						\forall x. {\bigwedge_{
								r {\in}\swform_1
							}} \varphi(\swform_1(\vec{X}^1),r) \leftrightarrow \varphi(\swform_2(\vec{Y}^1),r) 			
					\end{array}
					\!\!\!
					\right]			
			\end{array}
			\!\!\!\!
			\right]
		\end{array}
	\!\!\!\!
	\right]
\end{align*}

Intuitively, the formula describes that if there are $m$ \emph{distinct} sequences of sets of positions, described by the $X$'s that give the same value for $\varphi_1 \cond \prod_x \swform_1 : \nil$, then there are $m$ distinct sequences of sets, assigned to the $Y$'s that give that same value for $\varphi_2 \cond \prod_x \swform_2 : \nil$.

As Lemma \ref{lem:cw-less-than} demonstrates, formula $\cwform_1 \leq_l \cwform_2$ describes that, if the multisets returned by the formulas have size at most $l$, then $\cwform_2$ has all the elements of $\cwform_1$. Therefore, if both $\cwform_1 \leq_l \cwform_2$ and $\cwform_2 \leq_l \cwform_1$ are true for a string, then either the values of $\cwform_1$ and $ \cwform_2$ are too large, or they are the same.

\begin{lemma}\label{lem:cw-less-than}
	Let $l>0$, $\cwform_1= \sum_{\vec{X}} \varphi_1 
	\cond \prod_x \swform_1 
	$, and 
	$\cwform_2= \sum_{\vec{X}} \varphi_2 \cond \prod_x \swform_2 
	$.
	Then, $\Gamma \vdash \cwform_1 \leq_l \cwform_2$ if and only if 
	for every $(w,\sigma) \in \sat{\Gamma}$ and $\gamma \in R^{|w|}$,
	\[
		\sat{\cwform_2}(w,\sigma)(\gamma) \geq \min \{ l,~ \sat{\cwform_1}(w,\sigma)(\gamma) \}
		.
	\]
\end{lemma}

\begin{proof}
	We first observe, by Lemma \ref{lem:WtoMSO}, that  
	\[  
		\forall x. \bigwedge_{ r \in R(\swform)} \varphi(\swform,r) \leftrightarrow \varphi(\swform',r)
	\]
	is true at $(w,\sigma)$ if and only if $\sat{\swform}(w,\sigma) = \sat{\swform'}(w,\sigma)$.
	From the definition of $\cwform_1 \leq_l \cwform_2$ above, for every $(w,\sigma) \in \sat{\Gamma}$, $(w,\sigma) \in \sat{\cwform_1 \leq_l \cwform_2}$ exactly when, if there are $m \leq l$ (distinct) assignments to variables $\vec{X}$ for which $\varphi_1$ evaluates to true and $\swform_1$ returns a fixed value, then there are $m$ (respectively, distinct) assignments to variables $\vec{Y}$ for which $\varphi_2$ also evaluates to true and $\swform_2$ returns that same fixed value.
	By the completeness of \MSO, $\Gamma \vdash \cwform_1 \leq_l \cwform_2$ if and only if for every $(w,\sigma) \in \sat{\Gamma}$, $(w,\sigma) \in \sat{\cwform_1 \leq_l \cwform_2}$, and, by the above observation, the lemma follows. 
\end{proof}

The axioms for full \cwMSO\ include all the axioms for $\cwMSO(\cond,+)$, and, additionally, the ones in Table \ref{tab:cwmsofull-axioms}.
\begin{table}
	\centering
	\begin{tabular}{l l}
		\hline
    ($C11$): & $\begin{aligned}&\Gamma \vdash \cwform_1 \approx \cwform_2 \text{ implies } \Gamma \vdash \textstyle \sum_X \cwform_1 \approx \sum_X \cwform_2 \\[-1ex]
      &\text{if } X \text{ is not free in } \Gamma\end{aligned}$ \\[2ex]
		($C12$): & $\Gamma \vdash \sum_X \cwform \approx \sum_Y \cwform[Y / X]$ if $Y \notin \var(\cwform)$ \\[1ex]
		($C13$): & $\Gamma \vdash \sum_X \sum_Y \cwform \approx \sum_Y \sum_X \cwform$ \\[1ex]
		($C14$): & $\Gamma \vdash \sum_X (\cwform_1 {+} \cwform_2) \approx \sum_X \cwform_1 + \sum_X \cwform_2$ \\[1ex]
		($C15$): & $\Gamma \vdash \varphi \cond \sum_X \cwform_1 : \sum_X \cwform_2 \approx \sum_X \varphi \cond \cwform_2 : \cwform_1$ \\[1ex]
		($C16$): & $\Gamma \vdash \cwform \approx \sum_X \varphi \cond \cwform 
		$ if $\Gamma \vdash \exists ! X.~ \varphi(X)$ and $X \notin \var(\cwform)$ 
		\\[1ex]
    ($C17$): & $\Gamma \vdash \sum_{\vec{X}} \varphi_1 {\cond} {\prod_x} \swform_1 
    \approx \sum_{\vec{Y}} \varphi_2 {\cond} {\prod_x} \swform_2$ \\
    &if $\Gamma \vdash \cwform_1 \leq _l \cwform_2$ and $\Gamma \vdash \cwform_2 \leq _l \cwform_1$, \\
    &for $l = 2^{\ell( |\cwform_1| {+} |\cwform_2| {+} |\Gamma | ) \cdot \max \{ |\vec{X}|,|\vec{Y}| \}}$, \\
    &where $\cwform_1 = \sum_{\vec{X}} \varphi_1 \cond \prod_x \swform_1$ and $\cwform_2 = \sum_{\vec{X}} \varphi_2 \cond \prod_x \swform_2$. \\
		\hline \\
	\end{tabular}
	\caption{Axioms for $\cwMSO$.}
	\label{tab:cwmsofull-axioms}
\end{table}
The most interesting case is the one of Axiom (C17).
This axiom reduces proving the equivalence of the two sides to a bounded proof of their equivalence through \MSO.

\subsection{Soundness and Completeness}

We now prove that the axioms of Table \ref{tab:cwmsofull-axioms} are both sound and complete for \cwMSO.

\begin{lemma}\label{lem:merge-sums}
	For every $\Gamma$ and pair of \cwMSO\ formulas $\sum_{\vec{X}} \varphi_1 \cond \prod_x \swform_1$ and $\sum_{\vec{Y}} \varphi_2 \cond \prod_x \swform_2$, there is a \cwMSO\ formula $\sum_{\vec{Z}} \varphi_3 \cond \prod_x \swform_3$, such that 
	\begin{align*}
		\Gamma \vdash \sum_{\vec{X}} \varphi_1 \cond \prod_x \swform_1 + \sum_{\vec{Y}} \varphi_2 \cond \prod_x \swform_2 \approx 
		\sum_{\vec{Z}} \varphi_3 \cond \prod_x \swform_3.
	\end{align*}
\end{lemma}
\begin{proof}
	We first observe that $\Gamma \vdash \sum_{X} \nil \approx \nil$ --- a simple application of axiom (C16) and due to the completeness of \swMSO.
	We can assume, due to Axioms (C12), (C13), and (C16) that $\vec{X} = \vec{Y}$.
	Let $Z$ be a second-order variable that does not appear in any of the two given formulas, nor in $\Gamma$.
	We can see that 
	there are $\varphi'_1$ and $\varphi'_2$ that only have $Z$ as a free variable, such that 
	$\Gamma \vdash \exists ! Z. \varphi'_1 \land \exists ! Z. \varphi'_2 \land \forall Z. \neg (\varphi'_1 \land \varphi'_2)$ --- for instance, let $\varphi'_1 = \forall x.\neg Z(x)$ and $\varphi'_2 = \forall x. Z(x)$.
	
	We now observe that 
	\begin{align}
		 \Gamma &\vdash 
			\varphi_1' \cond \varphi_1 \cond \prod_x \swform_1 + 	 
			\varphi_2' \cond \varphi_2 \cond \prod_x \swform_2
			\approx \notag \\
			&\phantom{{}\vdash{}}
			(\varphi_1 \lor \varphi_2) \land (\varphi_1' \lor \varphi_2')
			\cond 
			\prod_x \varphi_1 \land \varphi_1' \cond \swform_1 : \swform_2.
			\label{eq:together}
	\end{align}
	Using the fact that $\varphi_1'$ and $\varphi_2'$ are mutually exclusive,
	by taking cases, we can see that the equation above is valid.
	Then, \eqref{eq:together} follows from the completeness of \swMSO.
	
	\begin{align*}
		\Gamma \vdash &\sum_{\vec{X}} \varphi_1 {\cond} {\prod_x} \swform_1 + \sum_{\vec{X}} \varphi_2 {\cond} {\prod_x} \swform_2 \approx 
		\tag{from (C16)}
		\\
		&
		\sum_{Z}\varphi_1' \cond \sum_{\vec{X}} \varphi_1 {\cond} {\prod_x} \swform_1 + \sum_{Z}\varphi_2'\cond \sum_{\vec{X}} \varphi_2 {\cond} {\prod_x} \swform_2 \approx
		\tag{from (C15) and (C14)}
		\\
		&
		\sum_{Z\vec{X}} 
		\varphi_1' \cond \varphi_1 \cond \prod_x \swform_1 + 	 
		\varphi_2' \cond \varphi_2 \cond \prod_x \swform_2
		\approx
		\tag{from \eqref{eq:together}}
		\\
%
%
		&	\sum_{Z\vec{X}} 
		(\varphi_1 {\lor} \varphi_2) {\land} (\varphi_1' {\lor} \varphi_2')
		\cond 
		\prod_x \varphi_1 {\land} \varphi_1' \cond \swform_1 : \swform_2.
		\tag*{\qedhere}
	\end{align*}
\end{proof}

\begin{definition}
	A $\cwMSO$ formula $\cwform$ is in first normal form if $\cwform$ is generated by the following grammar:
	\begin{align*}
		Q &::= \varphi \cond Q : Q \mid R \mid \nil;
		&R &::= R + R \mid S; \ \ \text{and}
		\\
		S &::= \sum_X S \mid \varphi \cond \prod_x \swform.
	\end{align*}
	It is in second normal form if 
	$+$ does not occur in $\cwform$.
\end{definition}

\begin{lemma}\label{lem:2nfs}
	For every $\Gamma$ and \cwMSO~ formula $\cwform$, there exists a \cwMSO~formula $\cwform'$ in second normal form, such that $\Gamma \vdash \cwform \approx \cwform'$.
\end{lemma}

\begin{proof}
	By Lemma \ref{lem:merge-sums}, it suffices to prove the lemma for 
	$\cwform'$ in first normal form.
	 We can see that axioms (C14) and (C12) allow us to push $+$ inside any sum operator, and (C15) and (C16) allow to do the same for conditionals. 
	 This gives us that $\Gamma \vdash \cwform \approx \sum_{\vec{X}}\cwform''$, where 
	 $\cwform''$ is a $\cwMSO(+,\cond)$ formula.
	 But then,
	 if  $\cwform_1, \cwform_2, \cwform_3$ are $\cwMSO(+,\cond)$ formulas, then 
	 \[ \varphi \cond (\cwform_1 + \cwform_2) : \cwform_3 \sim (\varphi \cond \cwform_1 : \cwform_3) + \varphi \cond\cwform_2, \]
	 gives that 
	 \[\Gamma \vdash \varphi \cond (\cwform_1 + \cwform_2) : \cwform_3 \approx (\varphi \cond \cwform_1 : \cwform_3) + \varphi \cond\cwform_2, \]
	 from the completeness of $\cwMSO(+,\cond)$.
	 Furthermore, it is not hard to see that, from the $\cwMSO(\cond,+)$ axioms,
	 \begin{align*}
	 	\Gamma \vdash &
	 	\varphi \cond \cwform_1 : \cwform_2 \approx \varphi \cond \cwform_1 + \neg \varphi \cond \cwform_2.
	 \end{align*}
	 Therefore, inside the sum operators of $\sum_{\vec{X}}\cwform''$, we can
	 bring all conditionals in the form $\varphi \cond \cwform_1$,
	 and then use axiom (C14) to 
	  eliminate all occurrences of $+$ inside the sum operators.
	 The remaining proof is similar to the proof of Lemma \ref{lem:normal-form}.
\end{proof}

\begin{theorem}[Completeness for $\cwMSO$]\label{thm:completeness-full}
	For every finite $\Gamma$ and $\cwMSO$ formulas $\cwform_1$ and $\cwform_1$, we have $\cwform_1 \sim_\Gamma \cwform_2$ if and only if $\Gamma \vdash \cwform_1 \approx \cwform_2$.
\end{theorem}

\begin{proof}
	\emph{The soundness} of the axioms is straightforward.
	The most interesting case is (C17), which we now prove sound.
	We assume that 
	$\Gamma \vdash \cwform_1 \leq _l \cwform_2$
	and
	$\Gamma \vdash \cwform_2 \leq _l \cwform_1$, 
	for 
$ l = 2^{\ell( |\cwform_1| {+} |\cwform_2| {+} |\Gamma | ) \cdot \max \{ |\vec{X}|,|\vec{Y}| \}}$, 
where
\[
\cwform_1 = 
\sum_{\vec{X}} \varphi_1 \cond \prod_x \swform_1 
\quad \text{and} \quad
\cwform_2 = 
\sum_{\vec{X}} \varphi_2 \cond \prod_x \swform_2 
.\]
Let $L = \ell( |\cwform_1| {+} |\cwform_2| {+} |\Gamma | ) $.
From Lemma \ref{lem:cw-less-than}, we get that
$\sat{\cwform_1}(w,\sigma)(\gamma) \geq \min \{ l,  \sat{\cwform_2}(w,\sigma)(\gamma)\}$
and 
$\sat{\cwform_2}(w,\sigma)(\gamma) \geq \min \{ l,  \sat{\cwform_1}(w,\sigma)(\gamma)\}$
for every $(w,\sigma)\in \sat{\Gamma}$ and $\gamma \in R^{|w|}$.

Let $(w,\sigma)\in \sat{\Gamma}$, where $|w| \leq L$.
By Corollary \ref{cor:finite-sat-Form}, it suffices to prove that 
$\sat{\cwform_1}(w,\sigma) = \sat{\cwform_2}(w,\sigma)$, and to do that, from the above discussion, it suffices to prove that 
for every $\gamma \in R^{|w|}$,
$l \geq \sat{\cwform_1}(w,\sigma)(\gamma)$ and $l \geq \sat{\cwform_2}(w,\sigma)(\gamma)$.
Specifically, we prove that 
$\sat{\cwform_1}(w,\sigma)(\gamma) \leq 2^{|w| \cdot |\vec{X}| }$ --- the case for $\cwform_2$ is symmetric.
The proof is by induction on $|\vec{X}|$: if $\cwform_1 = \varphi_1 \cond \prod_x \swform_1$, then it outputs at most one value, and therefore 
$\sat{\cwform_1}(w,\sigma)(\gamma) \leq 1 = 2^{L \cdot 0 }$;
and for the inductive step, 
$\sat{\sum_{Z}\cwform_1}(w,\sigma)(\gamma) \leq  
 2^{|w|}\sat{\cwform_1}(w,\sigma)(\gamma) \leq 
 2^{|w| + |w| \cdot |\vec{X}|} =  2^{|w| \cdot (|\vec{X}|+1)}$.
 
\emph{We now prove the completeness} of the axioms.
Let $\cwform_1 \sim_\Gamma \cwform_2$; we prove that
$\Gamma \vdash \cwform_1 \approx \cwform_2$.
By Lemma \ref{lem:2nfs}, we can assume that $\cwform_1$ and $\cwform_2$ are in second normal form. 
The proof is by induction on the total number of the top-level conditionals in these formulas. The inductive step is similar to the one in the proof of Theorem \ref{thm:completeness-cwmso}, so we only deal with the base cases.
If one of the formulas is $\nil$, then the other one is either $\nil$ or $\sum_{\vec{X}} \cwform$, where $\cwform \sim_\Gamma \nil$. By the completeness of $\cwMSO(\cond,+)$ and Axiom (C16), $\Gamma \vdash \nil \approx \sum_{\vec{X}} \cwform$ and we are done.
Finally, let $ \sum_{\vec{X}} \varphi_1 \cond \cwform_1 \sim_\Gamma \sum_{\vec{Y}} \varphi_2 \cond \cwform_2 $.
From axioms (C12), (C13), and (C16), we can assume that $\vec{X} = \vec{Y}$.
From
Lemma \ref{lem:cw-less-than}, 
\begin{align*}
\Gamma \vdash& \sum_{\vec{X}} \varphi_1 \cond \cwform_1 \leq_l \sum_{\vec{Y}} \varphi_2 \cond \cwform_2, \text{ and} \\
\Gamma \vdash& \sum_{\vec{X}} \varphi_2 \cond \cwform_2 \leq_l \sum_{\vec{Y}} \varphi_1 \cond \cwform_1, 
\end{align*} 
for every $l$.
Therefore, by using axiom (C17), 
\[
\Gamma \vdash \sum_{\vec{X}} \varphi_1 \cond \cwform_1 \approx \sum_{\vec{Y}} \varphi_2 \cond \cwform_2.
\qedhere 
\]
\end{proof}



\section{Equational Satisfiability is Undecidable}\label{sec:undec}

In this section, we prove that equational satisfiability for the full \cwMSO\
is undecidable.
A similar, but more complex construction can be made for $\cwFO$.
We first observe that, if we assume an unbounded set of values, the language of equations is closed under conjunction with respect to satisfiablity, in the sense of Lemma \ref{lem:closed-under-and}.
 
\begin{lemma}\label{lem:closed-under-and}
	Let $\cwform_1, \cwform_2, \cwform'_1, \cwform'_2 \in \cwMSO$ be such that $\cwform_1, \cwform_2$ use values that are distinct from the ones that 
$\cwform'_1, \cwform'_2$ use.
For every $w,\sigma$, 
$\sat{\cwform_1}(w,\sigma) = \sat{\cwform_2}(w,\sigma) $ and 
$\sat{\cwform'_1}(w,\sigma) = \sat{\cwform'_2}(w,\sigma) $, if and only if 
$\sat{\cwform_1+\cwform'_1}(w,\sigma) = \sat{\cwform_2+\cwform'_2}(w,\sigma) $.
\end{lemma}
\begin{proof}
  The lemma results 
  by
  observing 
  that the
  elements of the multisets $\sat{\cwform_1+\cwform'_1}(w,\sigma) $ and $ \sat{\cwform_2+\cwform'_2}(w,\sigma) $
  can be partitioned into those that use values that appear in $\cwform_1$ and $\cwform_2$,
  and those that use values that appear in $\cwform_1'$ and $\cwform_2'$.
\end{proof}


Fix a pair $(w,\sigma)$.
We use a series of formulas and equations to express that a pair $(w,\sigma)$ encodes the computation of a Turing Machine that halts.
Therefore, the question of whether there is such a pair that satisfies the resulting set of equations is undecidable.
Let $T = (Q,\Sigma,\delta,q_0,H)$ be a Turing Machine, where $Q$ is a finite set of states, $\Sigma$ is the set of symbols that the machine uses, $\delta: Q \times \Sigma \to Q \times \Sigma \times \{L,R\}$ is the machine's transition function, $q_0$ is the starting state, and $H$ is the halting state of $T$. 
We give the construction of the \cwMSO\ formula equations.

Let $\tend$ be a special symbol not in $\Sigma$.
A configuration of $T$ is represented by a string of the form $s_1 q s_2 \tend$, where $q$ is the current state for the configuration,  $s_1s_2$ is the string of symbols in the tape of the machine, and the head is located at the first symbol of $s_2$; $\tend$ marks the end of the configuration.
Let $x_0 \in \Sigma^+$ be an input of $T$ (for convenience, we assume that all inputs are nonempty). 

%
We use every $s \in Q \cup \Sigma \cup \{\tend\}$
as a predicate, so that $s(x)$ is true
if and only if the symbol $s$ is in position $x$.
%
We want to describe that $(w,\sigma)$ encodes a halting run of $T$ on $x_0$.
In other words, we must ensure that $(w,\sigma)$
is a sequence $c_0 \cdots c_k$ of configurations of $T$, such that
$c_0$ is $q_0 x_0 \tend $ and  $c_k$ is $s_1Hs_2 \tend$, where $s_1,s_2 \in \Sigma^*$.
 
We must therefore ensure that the following conditions hold:

\begin{enumerate}
	\item 
	$(w,\sigma)$ is of the form $c_0c_1\cdots c_k$, where each $c_i$  has exactly one $\tend$, at the end;
	\item 
	each $c_i$ is of the form $s_1 q s_2 \tend$, where $q \in Q$, $s_1s_2 \in \Sigma^*$, and $s_2 \neq \varepsilon$;
	\item $c_0 = q_0 x_0 \tend$;
	\item $c_k = s_1 H s_2 \tend$ for some $s_1,s_2$; and
	\item for every $0 \leq i < k$, $c_{i+1}$ results from $c_i$ by applying the transition function $\delta$. This condition can be further refined into the following subconditions. For every $0 \leq i < k$, if 
	$c_i =  x_1 ~x_2 \cdots ~x_{r} ~q_i ~y_{1} ~y_{2} \cdots ~y_{r'} \tend$, then:
	\begin{enumerate}
		\item if $\delta(q_i,y_{1}) = (q,x,L)$ and $r >0$, then $c_{i+1} = x_1 ~x_2 \cdots ~x_{r-1} ~q ~x_{r} ~x ~y_{2} \cdots ~y_{r'} \tend,$
		\item if $\delta(q_i,y_{1}) = (q,x,L)$ and $r =0$, then $c_{i+1} = q ~x ~y_2 \cdots ~y_{r'} \tend$, 
		\item if $\delta(q_i,y_{1}) = (q,x,R)$ and $r' > 1$, then $c_{i+1} = x_1 ~x_2 \cdots ~x_{r} ~x ~q ~y_{2} \cdots ~y_{r'} \tend$, and 
		\item if $\delta(q_i,y_{1}) = (q,x,R)$ and $r' = 1$, then $c_{i+1} = x_1 ~x_2 \cdots ~x_{r} ~x ~q ~\blank \tend$, where $\blank \in \Sigma$ is the symbol used by $T$ for a blank space.
	\end{enumerate}
\end{enumerate}



We now explain how to represent each of the conditions above with a formula or equation.
We use the following macros, where $0,1 \in R$ are two distinct weights:
\begin{align*}
	\nxt{x,y} &\defeq  (\neg (y \leq x)) \land \forall z.~ ( z \leq x \lor y \leq z ) \\
  \last{x} &\defeq \forall y. y{\leq} x 
	~\text{ and }
	\first{x} \defeq \forall y. y{\geq} x \\
	\firstc{x} &\defeq \first{x} \lor \exists y.~ \tend(y) \land \nxt{y,x} \\
	\valone{x} &\defeq \prod_y (x{=}y) {\cond} 1 
	~\text{ and }
  \valv{x}{s} \defeq \prod_y (x{=}y) {\cond} s 
  \\
	\firstcx{x,y} &\defeq \firstc{y} \land y \leq x \land \\
	&\phantom{{}\defeq{}}~~~~~~~\forall z . \neg (\tend(z) \land y \leq z \leq x)
		\\
	\pos{x}{v} &\defeq 
%
\sum_{X} \exists y.~ 
\firstcx{x,y} \\
&\phantom{\defeq}~~\land \forall z. (\neg z \in X) \lor (y \leq z \leq x)
\cond 
v 
\end{align*}

Intuitively, formula $\pos{x}{v}$ counts 
$2^i$, where $i$ is the position of $x$ in its configuration.
We note that
for each 
set $S$ of positions 
in a configuration, $S$ is uniquely identified by $\sum_{i \in S} 2^i$.
Furthermore, for each configuration, $\pos{x}{v}$ constructs a map from each symbol $s$ that appears in the set $S$ of positions (represented by the returned value $v$) to  $\sum_{i \in S} 2^i$. Therefore, the way that we will use $\pos{x}{v}$ (see how we deal with condition 5, below) gives a complete description of each configuration.

We now proceed to describe, for each of the conditions 1-6,
a number of equations that ensure that this condition holds.
By an equation, we mean something of the form $\cwform \eqeq \cwform'$,
where $\cwform$ and $\cwform'$ are $\cwMSO$ formulas.
Notice that by Lemma \ref{lem:embedMSOtoEq},
any \MSO\ formula can be turned into an equation
(as long as we have at least two distinct weights),
so for some conditions we give an \MSO\ formula rather than an equation.

A number of equations $\cwform_i \eqeq \cwform_i'$ ensures that the condition holds
in the sense that for any $(w,\sigma)$,
$\sat{\cwform_i}(w,\sigma) = \sat{\cwform_i'}(w,\sigma)$ for each $i$
if and only if $(w,\sigma)$ satisfies the condition.
By Lemma \ref{lem:closed-under-and},
once we have a number of equations $\cwform_i \eqeq \cwform_i'$
that together ensure that all conditions are satisfied,
the equation $\sum_i \cwform_i \eqeq \sum_i \cwform_i'$ ensures that all conditions are satisfied,
so that $(w,\sigma)$ satisfies the conditions if and only if
$\sat{\sum_i \cwform_i}(w,\sigma)= \sat{\sum_i \cwform_i'}(w,\sigma)$.
We omit most conditions, as it is not hard to express them in \FO, and only
demonstrate how to treat case a of condition 5. The other cases are analogous.

  Fix a transition $(q,s,q',s',L) \in \delta$ and $d \in \Sigma$.
  We use the following shorthand.
  \begin{align*}
    \tr{x,y,z} &\defeq q(y) \land y \leq x \land \forall y'.~ \neg (\tend(y') \land y \leq y' \leq x) \\
    &\phantom{{}\defeq{}}\land s(z) \land \nxt{y,z} \quad \text{and} \\
    \trprime{x,y,z} &\defeq q'(y) \land y \leq x \land \forall y'.~ \neg (\tend(y') \land y \leq y' \leq x) \\
    &\phantom{{}\defeq{}}\land s'(z) \land \nxt{y,z}.
  \end{align*}
  Let $s_1,s_2,\ldots,s_m$ be a permutation of $\Sigma$.
  We use the following equation:
  \begin{align*}
  \sum_{x}& \tend(x)
  \land
  \exists y. (\tend(y) \land x < y) 
  \land \\
  &\exists y,z. \tr{x,y,z}
    \cond 
    \sum_{y} y {\leq} x \land \forall z. (x {\leq} z \lor z {<} y \lor \neg \tend(x)) 
      \cond \\
      &q(y)
        \cond
        \pos{y}{\valv{x}{q}}: s_1(y)
          \cond
          \pos{y}{\valv{x}{s_1}}: \cdots : s_m(y)
            \cond
            \pos{y}{\valv{x}{s_m}} \\  \eqeq&
  \\
  \sum_{x}& \tend(x)
  \land
  \exists y. (\tend(y) \land x < y)
  \land \\
  &\exists y,z. \tr{x,y,z}
  \cond
  \sum_{y} x {\leq} y \land \forall z. (z {\leq} x \lor y {<} z \lor \neg \tend(x)) 
    \cond 
             \end{align*}
\begin{align*}
~~    &q'(y) \land \exists z. \nxt{y,z} \land s_1(z)
      \cond
      \pos{y}{\valv{x}{s_1}}: \cdots q'(y) \land \\
      &\exists z. \nxt{y,z} \land s_m(z)
        \cond
        \pos{y}{\valv{x}{s_m}}: \\
        &\exists z. q'(z) \land \nxt{z,y}
          \cond
          \pos{y}{\valv{x}{q}}: \\ 
          &\exists z,z'. q'(z) \land \nxt{z,z'} \land \nxt{z',y}
            \cond
            \pos{y}{\valv{x}{s}}: \\
            &s_1(y)
              \cond
              \pos{y}{\valv{x}{s_1}}: \cdots : s_m(y)
                \cond
                \pos{y}{\valv{x}{s_m}}
  \end{align*}
  The rightmost part of the equation ensures that if the effects of the transition are reversed, then all symbols are in the same place as in the previous configuration.
  We can then make sure that the state has changed to $q'$ and the symbol to $s'$ with the following formula:
  \begin{align*}
    \forall x,y.
    &\neg (\tend(x) \land \tend(y) \land \neg y \leq x \land \exists x_q,x_s. \tr{x,x_q,x_s}
    ) 
    \lor \\ 
    &\exists y_q,y_s. \trprime{y,y_q,y_s}.
  \end{align*}

\begin{theorem}\label{thm:sat-is-undec}
  If the set of weights has at least two distinct weights,
	the equational satisfiability problem for 
	\cwMSO\ and \cwFO\ is undecidable.
\end{theorem}

\begin{proof}
	For the case of \cwMSO, we use a reduction from the Halting Problem, as it is described above.
	It is not hard to see why conditions 1 to 5 suffice for the correctness of the reduction.
\end{proof}

\begin{remark}
	We note that Theorem \ref{thm:sat-is-undec} claims that the reduction works with at least two \emph{weights}, while our use of Lemma \ref{lem:closed-under-and} requires several different \emph{values}. 
	Weights are the elements of $R$, while values are possible outputs for the formulas, so they are multisets of strings of weights.
	If at least two weights are available, then it is not hard to see that any (finite) number of weights or values can be encoded in a string.
\end{remark}

\section{Conclusion}
We have given a sound and complete axiomatization for each of three
fragments of the weighted monadic second-order logic
\cwMSO.
Furthermore, we have investigated weighted versions of common decision problems for logics,
specifically model checking, satisfiability, and validity.
For the second layer of the logic, $\swMSO$,
these problems are all decidable,
although many of them have non-elementary complexity, inherited
from the corresponding problems for first- and second-order logic.
For the third layer, $\cwMSO$, 
we 
demonstrated 
that the problem of deciding
whether there exists an input that makes two given formulas return the same value is undecidable,
but 
deciding whether two formulas
return the same value for all inputs is 
decidable.

A natural open question of interest is to discover how different concrete semantics affect the decidability of equational satisfiability and validity.
As our results rely on the abstract semantics, one would hope to 
prove general requirements for the structure of concrete semantics, that 
would guarantee the decidability or undecidability of these problems.
Similarly, we hope for a modular way to give complete axiomatizations for concrete semantics, based on the axioms that we gave.
It would also be worthwhile to consider axiomatizing other useful relations of formulas, such as inequality. 
Finally, we observe that 
(C17) essentially reduces equational \cwMSO\ to MSO. Due to our layered axiomatizations, one can avoid using Axiom (C17), as long as the formulas do not have the sum operator.
However, 
we hope for an axiomatization that relies more on the syntax of the formulas.

\section*{Acknowledgments}
This work has been funded by the project ``Open Problems in the Equational Logic of Processes (OPEL)'' (grant no.~196050), the project ``Epistemic Logic for Distributed Runtime Monitoring'' (grant no.~184940), and the project ``MoVeMnt: Mode(l)s of Verification and Monitorability'' (grant no~217987) of the Icelandic Research Fund.

The authors are also thankful to the anonymous reviewers, whose comments have improved this paper.

\bibliographystyle{plain}
\bibliography{bibliography}

\appendix

\subsection{The Remaining Proof of Proposition \ref{prop:swmso-theorems}}

\begin{proof}
	\begin{enumerate}

		  \item[2) ] We have assumed $\Gamma \cup \{\varphi\} \vdash \swform_1 \approx \swform_2$,
		  and we get $\Gamma \cup \{\neg \varphi\} \vdash \swform_2 \approx \swform_2$ by reflexivity,
		  so $(S4)$ gives $\Gamma \vdash \varphi \cond \swform_1 : \swform_2 \approx \swform_2$.
		  Since we have assumed $\Gamma \vdash \varphi$, $(S3)$ gives
		  $\Gamma \vdash \varphi \cond \swform_1 : \swform_2 \approx \swform_1$.
		  Hence $\Gamma \vdash \swform_1 \approx \swform_2$ by symmetry and transitivity.
		  \item[3) ] By reflexivity, we have $\Gamma \vdash \swform \approx \swform$,
		  so ($S1$) gives both $\Gamma \cup \{\varphi\} \vdash \swform \approx \swform$
		  and $\Gamma \cup \{\neg \varphi\} \vdash \swform \approx \swform$,
		  so using ($S4$) we conclude $\Gamma \vdash \varphi \cond \swform : \swform \approx \swform$.
		
		
		    \item[6) ] Assume that $\Gamma \vdash \neg \varphi$.
		    By axiom ($S2$) we get $\Gamma \vdash \varphi \cond \swform_1 : \swform_2 \approx \neg \varphi \cond \swform_2 : \swform_1$,
		    and axiom ($S3$) gives $\Gamma \vdash \neg \varphi \cond \swform_2 : \swform_1 \approx \swform_2$,
		    so $\Gamma \vdash \varphi \cond \swform_1 : \swform_2 \approx \swform_2$.
		    \item[7) ] This is simply an instantiation of the fourth item of this proposition
		    where $\varphi_1 = \varphi_2 = \varphi$,
		    and the other two premises are guaranteed to hold because
		    $\{\varphi, \neg \varphi\}$ is inconsistent.
		    \item[8) ] Assume that $\Gamma \cup \{\varphi\} \vdash \swform_1 \approx \swform_2$
		    and $\Gamma \cup \{\neg \varphi\} \vdash \swform_1 \approx \swform_2$.
		    Then axiom (S4) gives
		    $\Gamma \vdash \varphi \cond \swform_1 : \swform_1 \approx \swform_2$,
		    and the third item of this proposition gives
		    $\Gamma \vdash \varphi \cond \swform_1 : \swform_1 \approx \swform_1$,
		    so $\Gamma \vdash \swform_1 \approx \swform_2$.
		    \item[9) ] Since $\Gamma \cup \{\varphi\} \vdash \varphi$,
		    we get $\Gamma \cup \{\varphi\} \vdash \varphi \cond \swform_1 : \swform_2 \approx \swform_1$
		    by ($S3$).
		\qedhere
	\end{enumerate}
\end{proof}

\subsection{The Proof of Theorem \ref{thm:completeness-swmso}}

\begin{proof}
	We show the soundness of each axiom in turn.
	\begin{description}
		\item[($S1$):]
		Assume that $\swform_1 \sim_\Gamma \swform_2$.
		Since $\sat{\Gamma \cup \{\varphi\}} = \sat{\Gamma} \cap \sat{\varphi}$,
		for any $(w,\sigma) \in \sat{\Gamma \cup \{\varphi\}}$
		we have $(w,\sigma) \in \sat{\Gamma}$,
		and hence $\sat{\swform_1}(w,\sigma) = \sat{\swform_2}(w,\sigma)$ by assumption.
		We conclude that $\swform_1 \sim_{\Gamma \cup \{\varphi\}} \swform_2$.
		\item[($S2$):]
		$\sat{\varphi \cond \swform_1 : \swform_2}(w,\sigma) = \sat{\swform_1}(w,\sigma)$ if and only if
		\[\sat{\neg \varphi \cond \swform_2 : \swform_1}(w,\sigma) = \sat{\swform_1}(w,\sigma),\]
		and likewise
		$\sat{\varphi \cond \swform_1 : \swform_2}(w,\sigma) = \sat{\swform_2}(w,\Sigma)$ if and only if
		$
		\sat{\neg \varphi \cond \swform_2 : \swform_1}(w,\sigma) = \sat{\swform_2}(w,\sigma).
		$
		It follows that
		$\varphi \cond \swform_1 : \swform_2 \sim_\Gamma \neg \varphi \cond \swform_2 : \swform_1$.
		\item[($S3$):]
		Assume $\Gamma \vdash \varphi$.
		By Corollary \ref{cor:completeness-mso},
		this means that $\Gamma \models \varphi$.
		Hence, for any $(w,\sigma) \in \sat{\Gamma}$
		we have $(w,\sigma) \models \varphi$,
		so $\varphi \cond \swform_1 : \swform_2 \sim_\Gamma \swform_1$.
		\item[($S4$):]
		Assume that $\swform \sim_{\Gamma \cup \{\varphi\}} \swform_1$ and 
		$\swform \sim_{\Gamma \cup \{\neg \varphi\}} \swform_2$.
		For every  $(w,\sigma) \in \sat{\Gamma}$, either $(w,\sigma) \in \sat{\varphi}$ or $(w,\sigma) \in \sat{ \neg \varphi}$.
		Therefore, for both cases,
		\[\sat{\varphi \cond \swform_1 : \swform_2}(w,\sigma) = \sat{\swform}(w,\sigma),\]
		and we conclude that 
		$\varphi \cond \swform_1 : \swform_2 \sim_\Gamma  \swform$.
		\qedhere
	\end{description}
\end{proof}

\subsection{The Proof of Theorem \ref{thm:completeness-cwmso}}

\begin{proof}
	We show the soundness of each axiom in turn.
	
	Axiom ($C1$):
	\begin{align*}
		\sat{\cwform + \nil}(w,\sigma) &= \sat{\cwform}(w,\sigma) \munion \sat{\nil}(w,\sigma) \\
		&= \sat{\cwform}(w,\sigma) \munion \emptyset = \sat{\cwform}(w,\sigma).
	\end{align*}
	
	Axiom ($C2$):
	\begin{align*}
		\sat{\cwform_1 + \cwform_2}(w,\sigma) &= \sat{\cwform_1}(w,\sigma) \munion \sat{\cwform_2}(w,\sigma) \\
		&= \sat{\cwform_2}(w,\sigma) \munion \sat{\cwform_1}(w,\sigma) \\
		&= \sat{\cwform_2 + \cwform_1}(w,\sigma).
	\end{align*}
	
	Axiom ($C3$):
	\begin{align*}
		\phantom{{}={}} &\sat{(\cwform_1 + \cwform_2) + \cwform_3}(w,\sigma) \\
		= &\sat{(\cwform_1 + \cwform_2)}(w,\sigma) \munion \sat{\cwform_3}(w,\sigma) \\
		= &(\sat{\cwform_1}(w,\sigma) \munion \sat{\cwform_2(w,\sigma)}) \munion \sat{\cwform_3}(w,\sigma) \\
		= &\sat{\cwform_1}(w,\sigma) \munion (\sat{\cwform_2}(w,\sigma) \munion \sat{\cwform_3}(w,\sigma)) \\
		= &\sat{\cwform_1}(w,\sigma) \munion \sat{\cwform_2 + \cwform_3}(w,\sigma) \\
		= &\sat{\cwform_1 + (\cwform_2 + \cwform_3)}(w,\sigma).
	\end{align*}
	
	Axiom ($C4$): This follows from soundness of \swMSO\ and Lemma \ref{lem:forall}.
	
	Axiom ($C5$): If $y \notin \var(\swform)$, then
	\begin{align*}
		\sat{{\textstyle\prod_x} \swform}(w,\sigma) &= \multiset{\sat{\swform}(w,\sigma[x \mapsto 1] \dots \sat{\swform}(w,\sigma[x \mapsto |w|)} \\
		&= \multiset{\sat{\swform[y/x]}(w,\sigma[y \mapsto 1]) \\ &\phantom{{}={}}\cdots \sat{\swform[y/x]}(w,\sigma[y \mapsto |w|])} \\
		&= \sat{{\textstyle \prod_y} \swform[y/x]}(w,\sigma).
	\end{align*}
	
	Axioms ($C6$)--($C9$): The proof of these is similar to
	the corresponding proofs in Theorem \ref{thm:completeness-swmso}.
	
	Axiom ($C10$): We evaluate by cases. If $(w,\sigma) \models \varphi$, then
	\begin{align*}
		\phantom{{}={}} &\sat{(\varphi \cond \cwform' : \cwform'') + \cwform}(w,\sigma) \\ = &\sat{(\varphi \cond \cwform' : \cwform'')}(w,\sigma) \munion \sat{\cwform}(w,\sigma) \\
		= &\sat{\cwform'}(w,\sigma) \munion \sat{\cwform}(w,\sigma) \\
		= &\sat{\cwform' + \cwform}(w,\sigma).
	\end{align*}
	and
	\begin{align*}
		\sat{\varphi \cond (\cwform' + \cwform) : (\cwform'' + \cwform)}(w,\sigma) &= \sat{\cwform' + \cwform}(w,\sigma)
	\end{align*}
	Likewise, if $(w,\sigma) \models \neg \varphi$, then
	\begin{align*}
		\phantom{{}={}} &\sat{(\varphi \cond \cwform' : \cwform'') + \cwform}(w,\sigma) \\ = &\sat{(\varphi \cond \cwform' : \cwform'')}(w,\sigma) \munion \sat{\cwform}(w,\sigma) \\
		= &\sat{\cwform''}(w,\sigma) \munion \sat{\cwform}(w,\sigma) \\
		= &\sat{\cwform'' + \cwform}(w,\sigma)
	\end{align*}
	and
	\begin{align*}
		\sat{\varphi \cond (\cwform' + \cwform) : (\cwform'' + \cwform)}(w,\sigma) &= \sat{\cwform'' + \cwform}(w,\sigma).
	\end{align*}
	
	For completeness,
	assume $\cwform_1 \sim_\Gamma \cwform_2$.
	By Lemma \ref{lem:normal-form}, there exist formulas $\cwform_1'$ and $\cwform_2'$,
	both in normal form, such that $\Gamma \vdash \cwform_1 \approx \cwform_1'$
	and $\Gamma \vdash \cwform_2 \approx \cwform_2'$.
	By soundness,
	this implies $\cwform_1 \sim_\Gamma \cwform_1'$ and $\cwform_2 \sim_\Gamma \cwform_2'$,
	so $\cwform_1' \sim_\Gamma \cwform_2'$.
	Since these are in normal form, Lemma \ref{lem:nf-completeness}
	gives $\Gamma \vdash \cwform_1' \approx \cwform_2'$,
	and by symmetry and transitivity, this implies $\Gamma \vdash \cwform_1 \approx \cwform_2$.
\end{proof}

\subsection{The Full Proof for the Undecidability of Equational Satisfiability}

We now present the full construction of the reduction that proves that equational satisfiability of \cwFO\ (and therefore also of \cwMSO) is undecidable.
We take special care to only use \cwFO\ formulas, and therefore we use a special construction for recording the positions where each symbol appears in a configuration.

Fix a pair $(w,\sigma)$.
We use a series of formulas and equations to express that a $(w,\sigma)$ encodes the computation of a Turing Machine that halts.
Therefore, the question of whether there is such a pair that satisfies the resulting set of equations is undecidable.
Let $T = (Q,\Sigma,\delta,q_0,H)$ be a Turing Machine, where $Q$ is a finite set of states, $\Sigma$ is the set of symbols that the machine uses, $\delta: Q \times \Sigma \to Q \times \Sigma \times \{L,R\}$ is the machine's transition function, $q_0$ is the starting state, and $H$ is the halting state of $T$. 
Let $\tend, \mi, \cnt$ be special symbols not in $\Sigma$.
A configuration of $T$ is represented by a string of the form $s_1 q s_2 \tend$, where $q$ is the current state for the configuration,  $s_1s_2$ is the string of symbols in the tape of the machine, and the head is located at the first symbol of $s_2$; $\tend$ marks the end of the configuration.
Let $x_0 \in \Sigma^*$ be an input of $T$. 


We use every $s \in Q \cup \Sigma \cup \{\tend, \cnt, \mi \}$
as a predicate, so that $s(x)$ is true
if and only if the symbol $s$ is in position $x$.
Let $[0] = \cnt$, and for every $i \geq 1$, let $[i] = \cnt^{2^{i-1}} \mi \cnt^{2^{i-1}}$, so that in $[i]$, $\cnt$ appears exactly $2^{i}$ times.
Then, for every string $y_0y_1\cdots y_j \in (Q \cup \Sigma \cup \{\tend \})^j$, let $[y_0y_1\cdots y_j] = [0]y_0[1]y_1\cdots [j]y_j$.
%
We want to describe that $(w,\sigma)$ encodes a halting run of $T$ on $x_0$.
In other words, we must ensure that $(w,\sigma)$
is $[c_0] \cdots [c_k]$, where $c_0 \cdots c_k$ is a sequence of configurations of $T$, such that
$c_0$ is $q_0 x_0 \tend $ and  $c_k$ is $s_1Hs_2 \tend$, where $s_1,s_2 \in \Sigma^*$.

We must therefore ensure that the following conditions hold:

\begin{enumerate}
	\item 
	$(w,\sigma)$ is of the form $[c_0][c_1]\cdots [c_k]$, where each $c_i$  has exactly one $\tend$, at the end;
	\item 
	each $c_i$ is of the form $s_1 q s_2 \tend$, where $q \in Q$, $s_1s_2 \in \Sigma^*$, and $s_2 \neq \varepsilon$;
	\item $c_0 = q_0 x_0 \tend$;
	\item $c_k = s_1 H s_2 \tend$ for some $s_1,s_2$; and
	\item for every $0 \leq i < k$, $c_{i+1}$ results from $c_i$ by applying the transition function $\delta$. This condition can be further refined into the following subconditions. For every $0 \leq i < k$, if 
	$c_i =  x_1 ~x_2 \cdots ~x_{r} ~q_i ~y_{1} ~y_{2} \cdots ~y_{r'} \tend$, then:
	\begin{enumerate}
		\item if $\delta(q_i,y_{1}) = (q,x,L)$ and $r >0$, then $c_{i+1} = x_1 ~x_2 \cdots ~x_{r-1} ~q ~x_{r} ~x ~y_{2} \cdots ~y_{r'} \tend,$
		\item if $\delta(q_i,y_{1}) = (q,x,L)$ and $r =0$, then $c_{i+1} = q ~x ~y_2 \cdots ~y_{r'} \tend$, 
		\item if $\delta(q_i,y_{1}) = (q,x,R)$ and $r' > 1$, then $c_{i+1} = x_1 ~x_2 \cdots ~x_{r} ~x ~q ~y_{2} \cdots ~y_{r'} \tend$, and 
		\item if $\delta(q_i,y_{1}) = (q,x,R)$ and $r' = 1$, then $c_{i+1} = x_1 ~x_2 \cdots ~x_{r} ~x ~q ~\blank \tend$, where $\blank \in \Sigma$ is the symbol used by $T$ for a blank space.
	\end{enumerate}
\end{enumerate}



We now explain how to represent each of the conditions above with a formula or equation.
We use the following macros, where $0,1 \in R$ are two distinct weights:
\begin{align*}
\sym{x} &\defeq \bigvee_{s \in Q \cup \Sigma \cup \{\tend \}} s(x) \\
\nxt{x,y} &\defeq  (\neg (y \leq x)) \land \forall z.~ ( z \leq x \lor y \leq z ) \\
\first{x} &\defeq \forall y.~ y\geq x \\
\firstc{x} &\defeq \first{x} \lor \exists y.~ \tend(y) \land \nxt{y,x} \\
\last{x} &\defeq \forall y.~ y\leq x \\
\nxtsymb{x,y} &\defeq  (\neg y \leq x) \land \forall z.~ (\neg \sym{z} \lor z \leq x \lor y \leq z ) \\
\valone{x} &\defeq \prod_y (x=y) \cond 1 : 0 \\
\valv{x}{s} &\defeq \prod_y (x=y) \cond s : 0 \\
\pos{x}{v} &\defeq \forall y.~ \neg \sym{y} \lor x \leq y \cond v : \\
&\phantom{{}\defeq{}}\sum_{y} \exists z.~ \sym{z} \land \nxtsymb{z,x} \land \\
&\phantom{{}\defeq{}}z \leq y \leq x \land \cnt(y) \cond v : v_0
\end{align*}

Intuitively, formula $\pos{x}{v}$ counts how many $\cnt$s appear right before position $x$.
We note that, as long as condition 1 is satisfied, for each symbol $s$ that appears in the set $S$ of positions in a configuration, $S$ is uniquely identified by $\sum_{i \in S} 2^i$.
Furthermore, for each configuration, $\pos{x}{v}$ constructs a map from each such $s$ (represented by the returned value $v$) to  $\sum_{i \in S} 2^i$. Therefore, the way that we will use $\pos{x}{v}$ (see how we deal with condition 5, below) gives a complete description of each configuration.

We will use \valzer\ as the default (negative) value in conditionals,
and as such $\varphi \cond v$ is used as shorthand for $\varphi \cond v : \valzer$.
Furthermore, we assume that $:$ binds to the nearest $\cond$, and therefore, $\varphi_1 \cond \varphi_2 \cond \cwform_1 : \cwform_2$ means $\varphi_1 \cond \varphi_2 \cond \cwform_1 : \cwform_2 : \valzer$, which can be uniquely parsed as $\varphi_1 \cond (\varphi_2 \cond \cwform_1 : \cwform_2) : \valzer$.

We now proceed to describe, for each of the conditions 1-6,
a number of equations that ensure that this condition holds.
By an equation, we mean something of the form $\cwform \eqeq \cwform'$,
where $\cwform$ and $\cwform'$ are $\cwFO$ formulas.
Notice that by Lemma \ref{lem:embedMSOtoEq},
any first-order formula can be turned into an equation
(as long as we have at least three distinct weights),
so for some conditions we give a first-order formula rather than an equation.

A number of equations $\cwform_i \eqeq \cwform_i'$ ensures that the condition holds
in the sense that for any $(w,\sigma)$,
$\sat{\cwform_i}(w,\sigma) = \sat{\cwform_i'}(w,\sigma)$ for each $i$
if and only if $(w,\sigma)$ satisfies the condition.
By Lemma \ref{lem:closed-under-and},
once we have a number of equations $\cwform_i \eqeq \cwform_i'$
that together ensure that all conditions are satisfied,
the equation $\sum_i \cwform_i \eqeq \sum_i \cwform_i'$ ensures that all conditions are satisfied,
so that $(w,\sigma)$ satisfies the conditions if and only if
$\sat{\sum_i \cwform_i}(w,\sigma)= \sat{\sum_i \cwform_i'}(w,\sigma)$.

\begin{enumerate}
	\item 
	We describe this condition using a first order formula and two equations.
	The formula makes sure that the word is of the form 
	$d_0d_1\cdots d_k$, where each $d_i$ is of the form 
	$\cnt y_0\cnt^{n_1}\mi\cnt^{n_2}y_1\cdots \cnt^{n_{K-1}} \mi \cnt^{n_K}y_K$, where $n_1 \cdots n_K$ is a sequence of non-negative integers and $y_K = \tend$:
	\begin{align*}
  	& (\forall x.~ \neg \firstc{x} \lor (\cnt(x) \land \exists y.~  \nxt{x,y} \land \sym{y} )) \\
  	&\land (\exists x.~ \last{x} \land \tend(x)) \\
  	&\land \forall x.~ \neg \sym{x} \lor \last{x} \\
    &\lor \exists y,z.~ \nxtsymb{x,z} \land x \leq y \leq z \land \mi(y) \\ 
  	&\land \forall i.~ i \leq x \lor z \leq i \lor i = y \lor \cnt(i) .
	\end{align*}
	The following equation ensures that the same number of $\cnt$'s appear before and after $\mi$:
	\begin{align*}
  	&\sum_{x} \mi(x) \cond \sum_{y} \exists x',y'.~ x' \leq y \leq x \leq y' \\
    &\land \nxtsymb{x',y'} \land \cnt(y) \cond \valone{x} \\
  	&\eqeq \\
  	&\sum_{x} \mi(x) \cond \sum_{y} \exists x',y'.~ x' \leq x \leq y \leq y' \\
    &\land \nxtsymb{x',y'} \land \cnt(y) \cond  \valone{x} 
	\end{align*}
	Finally, in the context of the formula and equation above, the following equation ensures that for every $1 \leq i \leq K$, $n_i = 2^{i}$:
	\begin{align*}
  	&\sum_{x} \sym{x} \land \neg \last{x} \cond \\
  	&(\forall y.~ \neg \sym{y} \lor x \leq y \cond \valone{x}) : \\ 
  	&\sum_{y} \exists z.~ \sym{z} \land \nxtsymb{z,x} \\
    &\land z \leq y \leq x \land \cnt(y) \cond \valone{x}) \\
  	&\eqeq \\
    &\sum_{x} \sym{x} \land \neg \last{x} \cond \\
    &\sum_{y} \exists z.~ \mi(z) \land \nxtsymb{x,z} \\
    &\land x \leq y \leq z \land \cnt(y) \cond \valone{x}
	\end{align*}
	\item 
	For this condition, it suffices to require that between each pair of state symbols, there is a $\tend$ symbol, and between two occurrences of $\tend$, there is a state symbol, and right after each state symbol, there is a symbol from the alphabet. The following first-order formula expresses this:
	\begin{align*}
  	&\forall x,y.~ \neg \tend(x) \lor \neg \tend(y) \lor \neg x \leq y \\
    &~~~~\lor \exists z.~ x \leq z \leq y \land \bigvee_{q \in Q} q(z) \\
  	&\land \forall x,y.~ \neg \bigvee_{q \in Q} q(x) \lor \neg \bigvee_{q \in Q} q(y) \\
    &~~~~\lor \neg x \leq y  \lor \exists z.~ x \leq z \leq y \land \tend(z) \\
  	& \land \forall x.~ \neg  \bigvee_{q \in Q} q(x) \lor \exists y.~ \sym{y} \\
    &~~~~\land \nxtsymb{x,y} \land \neg \tend(y)
	\end{align*}
	\item This condition can be imposed by a first order formula that explicitly describes $c_0$.
	\item By the first-order formula 
	$\exists x.~ H(x) \land \forall y.~ \neg \tend(y) \lor \neg y \geq x \lor \last{y}$.
	\item 
	We demonstrate how to treat case a. The other cases are analogous. 
	Fix a transition $(q,s,q's',L) \in \delta$ and $d \in \Sigma$.
	We use the following shorthand.
	\begin{align*}
    \tr{x,y,z} &\defeq q(y) \land y \leq x \\
    &\phantom{{}\defeq{}}\land \forall y'.~ \neg (\tend(y') \land y \leq y' \leq x) \\
    &\phantom{{}\defeq{}}\land s(z) \land \nxtsymb{y,z} \quad \text{and} \\
    \trprime{x,y,z} &\defeq q'(y) \land y \leq x \\
    &\phantom{{}\defeq{}}\land \forall y'.~ \neg (\tend(y') \land y \leq y' \leq x) \\ &\phantom{{}\defeq{}}\land s'(z) \land \nxtsymb{y,z}
  \end{align*}
	Let $s_1,s_2,\ldots,s_m$ be a permutation of $\Sigma$.
	We use the following equation:
	\begin{align*}
	&\sum_{x} \tend(x)
  {\land}
  \exists y. (\tend(y) {\land} \neg y {\leq} x )
	{\land} 
	\exists y,z. \tr{x,y,z}
	{\cond} \\
    &~~\sum_{y} \sym{y} {\land} y {\leq} x {\land} \forall z. (x {\leq} z \lor \neg y {\leq} z \lor \neg \tend(x)) 
    {\cond} \\
      &~~~~q(y)
      {\cond}
        \pos{y}{\valv{x}{q}}:	s_1(y)
        {\cond}
          \pos{y}{\valv{x}{s_1}}: s_2(y)
          {\cond}
            \pos{y}{\valv{x}{s_2}}: \\
            &~~~~\cdots : s_m(y)
            {\cond}
              \pos{y}{\valv{x}{s_m}} \\
  &\eqeq \\
	&\sum_{x} \tend(x)
  {\land}
  \exists y. (\tend(y) {\land} \neg y {\leq} x)
  {\land}
  \exists y,z. \tr{x,y,z}
	{\cond} \\
    &~\sum_{y} \sym{y} {\land} x {\leq} y {\land} \forall z. (z {\leq} x \lor \neg z {\leq} y \lor \neg \tend(x)) 
    {\cond} \\
    &~q'(y) {\land} \exists z. \sym{z} {\land} \nxtsymb{y,z} {\land} s_1(z)
      {\cond}
      \pos{y}{\valv{x}{s_1}} {:} \\
      &~q'(y) {\land} \exists z. \sym{z} {\land} \nxtsymb{y,z} {\land} s_2(z)
        {\cond}
        {\pos{y}{\valv{x}{s_2}}} {:} {\cdots} \\
        &~q'(y) {\land} \exists z. \sym{z} {\land} \nxtsymb{y,z} {\land} s_m(z)
          {\cond}
          \pos{y}{\valv{x}{s_m}} {:} \\
          &~\exists z. q'(z) {\land} \nxtsymb{z,y}
            {\cond}
            \pos{y}{\valv{x}{q}} {:} \\
            &~\exists z,z'. q'(z) {\land} \nxtsymb{z,z'} {\land} \nxtsymb{z',y} 
              {\cond}
              \pos{y}{\valv{x}{s}} {:} \\
              &~s_1(y)
                {\cond}
                \pos{y}{\valv{x}{s_1}} {:} s_2(y)
                  {\cond}
                  \pos{y}{\valv{x}{s_2}} {:} \cdots {:} s_m(y) {\cond} \pos{y}{\valv{x}{s_m}}
	\end{align*}
	The rightmost part of the equation ensures that if the effects of the transition are reversed, then all symbols are in the same place as in the previous configuration.
	We can then make sure that the state has changed to $q'$ and the symbol to $s'$ with the following formula:
	\begin{align*}
  	\forall x,y. &\neg (\tend(x) \land \tend(y) \land \neg y {\leq} x \land \exists x_q,x_s. \tr{x,x_q,x_s}) \\
    &\lor \exists y_q,y_s.~ \trprime{y,y_q,y_s}.
	\end{align*}
\end{enumerate}


\begin{proof}[Proof of Theorem \ref{thm:sat-is-undec}]
	We use a reduction from the Halting Problem, as it is described above.
	It is not hard to see why conditions 1 to 5 suffice for the correctness of the reduction, and it is not hard to see that the formulas we construct ensure the corresponding conditions.
	Furthermore, notice that all formulas are \cwFO\ formulas, and therefore the problem is undecidable for \cwFO, but also for \cwMSO, which is a more general case.
\end{proof}

\end{document}